\documentclass[11pt]{article}
\pdfoutput=1
\usepackage{fullpage}
\usepackage{graphicx}
\usepackage{subcaption}
\usepackage[misc,geometry]{ifsym}
\usepackage{color}
\usepackage{amsmath,amssymb,amsthm,amsfonts,bm}
\usepackage{setspace}

\newtheorem{theorem}{Theorem}
\newtheorem{lemma}{Lemma}
\newtheorem{proposition}{Proposition}
\newtheorem{definition}{Definition}
\newtheorem{corollary}{Corollary}

\begin{document}
\title{When does additional information lead to longer travel time in multi-origin–destination networks?\thanks{Xujin Chen and Xinqi Jing were supported by the Chinese Academy of Sciences [Grants XDA27000000 and ZDBS-LY-7008] and National Natural Science Foundation of China [Grant No. 12331014]. Zhongzheng Tang was supported by National Natural Science Foundation of China [Grant No. 12101069].}}

\author{Xujin Chen\footnote{Academy of Mathematics and Systems Science, Chinese Academy of Sciences, Beijing 100190, China; \textsf{xchen@amss.ac.cn}} \and  
Xiaodong Hu\footnote{Academy of Mathematics and Systems Science, Chinese Academy of Sciences, Beijing 100190, China; \textsf{xdhu@amss.ac.cn}} \and
Xinqi Jing\footnote{Academy of Mathematics and Systems Science, Chinese Academy of Sciences, Beijing 100190, China; \textsf{jingxinqi@amss.ac.cn}} \and
Zhongzheng Tang\footnote{School of Mathematical Sciences, Beijing University of Posts and Telecommunications, Beijing 100876, China; \textsf{tangzhongzheng@amss.ac.cn}}
}
\date{}
\maketitle              
\begin{abstract}
The Informational Braess' Paradox (IBP) illustrates a counterintuitive scenario where revelation of additional roadway segments to \emph{some} self-interested travelers leads to increased travel times for \emph{these} individuals.  IBP extends the original Braess' paradox by relaxing the assumption that all travelers have identical and complete information about the network. 
In this paper, we study the conditions under which IBP does not occur in networks with non-atomic selfish travelers and multiple origin-destination pairs. Our results completely characterize the network topologies immune to IBP, thus resolving an open question proposed by Acemoglu et al.\ \cite{ibp18}.

\end{abstract}
\section{Introduction}\label{sec:intro}

The \emph{informational Braess' paradox} (IBP), as introduced by Acemoglu et al.\ \cite{ibp18}, illustrates a counterintuitive phenomenon in network traffic where disseminating supplementary information about available paths to \emph{certain} travelers paradoxically results in increased travel time for \emph{these} individuals. This longer travel time stems from the initial reactions of the informed travelers to the new information, and the subsequent responsive adjustments in routing by those without information extension. The IBP concept builds upon the classical Braess' paradox (BP), which demonstrates that adding an extra link to a transport network can sometimes lead to a situation where all travelers experience a weak increase in travel times as a result of their individually rational routing decisions. The ``weak increase'' in this context means that travel times for travelers between one origin-destination (OD) pair strictly increase, while travel times for all others do not decrease \cite{braess1968,chen2016network}.

Both IBP and BP reveal the complex and counterintuitive interplay between network structure, the information possessed by self-interested travelers, and the resulting traffic flow dynamics.  Though both paradoxes show similar negative effects that arise from seemingly beneficial changes, they each highlight different aspects of network dynamics and the potential unforeseen consequences when adjustments are made. (1) In BP, all travelers have access to complete information about the available routes in the network, which leads to uniform decision-making, ironically resulting in suboptimal outcomes when a new link is introduced. Conversely, IBP addresses the effects of unequal information distribution among travelers, where individual travelers act based on typically incomplete and differentiated information. In this scenario, travelers belong to different groups determined by their information types, i.e., the information sets about network edges available to them. When only one group of travelers receives additional information, their initial responses to this new data trigger a ripple effect that can ultimately impact all travelers. (2) BP requires that the introduction of an extra link causes a ``uniform'', albeit weak, increase in travel times for all travelers. On the other hand, IBP only stipulates that travelers who receive additional information experience longer travel time. It is possible that some travelers with unchanged information sets might benefit from shorter travel times. As pointed out by Acemoglu et al.\ \cite{ibp18}, in networks with a single OD pair, BP is identical with IBP with a single information type. The observation, however, cannot be extended to the multi-OD-pair case. The BP with two OD pairs, presented in \cite{chen2016network}, cannot be formulated as an IBP with two information types. Indeed, in the BP only one group of travelers alters their routes following the addition of an extra edge, yet their travel time does not increase. The lack of increased travel time for the route-changing group in BP departs from the expected outcome in an IBP framework, where travelers with additional choices face longer travel times after reacting to new information.

Since its introduction in the 1960s \cite{braess1968}, BP has been extensively studied from various perspectives. Notably, for non-atomic travelers, each of whom controls an infinitesimally small fraction of the total traffic, Milchtaich \cite{milchtaich2006network} and Chen et al. \cite{chen2016network} characterized topologies of networks that are immune to BP, thereby advancing our understanding of when and why the paradox does not manifest. In contrast, the investigation into IBP is still in its nascent stages. The body of research on IBP is considerably smaller, indicating a promising avenue for study to explore the nuances and implications of information asymmetry in network traffic dynamics. Acemoglu et al. \cite{ibp18} proved that IBP with a single OD pair does not occur if and only if the network is series of linearly independent (SLI). The result, in combination with Milchtaich's series-parallel characterization \cite{milchtaich2006network}, indicates that the IBP is a considerably more pervasive phenomenon than BP, as SLI networks form a subclass of series-parallel ones \cite{ibp18}. The conclusion aligns with the aforementioned observation \cite{ibp18} that IBP generalizes BP when restricted to single OD pair networks. However, as far as multiple OD pairs are concerned, the relation between IBP and BP becomes considerably more intricate, as highlighted in the preceding discussion. Can one expect a smaller class of networks for IBP-freeness than those, obtained in \cite{chen2016network}, for BP-freeness (as what has been done for the single OD pair case)? The answer is negative, as shown by Acemoglu et al.\ \cite{ibp18}. The authors provided a sufficient condition that precludes the occurrence of IBP in networks with multiple OD pairs, while also demonstrating that this condition is not necessary. They thereby posed the natural question of identifying a condition that is both sufficient and necessary. In this paper, we resolve this question by establishing a complete characterization of undirected networks with multiple OD pairs in which IBP never occurs for non-atomic travelers. 
The characterization might be helpful in network design to improve traffic efficiency.

\subsection{Routing Game}
Our basic  framework  is {\em selfish routing}  model \cite{czumaj2004,koutsoupias2009,roughgarden2002}. In this paradigm, travelers, acting in their own self-interest, select routes connecting their origins to their destinations within a traffic network. This network is represented by an undirected connected multigraph  $G=(V,E)$ that may have parallel edges but no loops, where $V=V(G)$ is the vertex set and $E=E(G)$ is the edge set. If no confusion arises, we identify a graph with its edge set. Given $n$ pairs of origin-destination (OD) vertices $(o_i,d_i)_{i=1}^n$ in $G$ with $o_i\ne d_i$ for all $i\in[n]$, we often call them OD pairs. There are $k$ types of \emph{non-atomic} travelers, who are individually insignificant relative to the entire system. The ``non-atomic'' attribute means that each traveler controls an infinitesimally small fraction of the total traffic. No single non-atomic traveler's actions can influence the overall system; their contribution to congestion is negligible. For every $j\in[k]$, the amount of type-$j$ travelers is often referred to as their {\em traffic rate}, and denoted as $r(j)$ or $r_j$  (we will use these notations interchangeably throughout). 
Type-$j$ travelers have a common information set $E_j\subseteq E$ and a common terminal set $\{o_{t(j)},d_{t(j)}\}$, where $t(\cdot)$ is a function from $[k]$ to $[n]$. Each type-$j$ traveler travels along a path from his origin $o_{t(j)}$ to his destination $d_{t(j)}$ in $E_j$, called an $o_{t(j)}$-$d_{t(j)}$ path, or simply his \emph{OD path}. Let $\mathcal{P}_j$ denote the set of type-$j$ travelers' OD paths.
The routing, a.k.a. traffic \emph{flow}, of type-$j$ travelers is denoted by $f^j=(f^j_P:P\in\mathcal P_j)$, 
where $f^j_P\ge0$ equals the volume of the flow (the amount of  type-$j$ travelers) going through path $P\in\mathcal P_j$, and the flow volume  $\sum_{P\in\mathcal P_j}f^j_P$ equals the traffic rate $r_j$. 
The routing of all $k$ types of travelers is represented by a \emph{flow} (vector) $f=(f^1,\ldots,f^k)$. 
For any edge $e\in E$, let $f_e^j=\sum_{P\in\mathcal P_j:e\in P}f^j_P$ denote the flow volume of type-$j$ travelers traversing $e$, and $f_e=\sum_{j\in [k]}f_e^j$ denote the flow volume of all travelers passing through $e$. For any path $P$ in $G$, 
 let $f_P=\sum_{j\in [k]}f_P^j$ denote the flow volume of all travelers going through $P$. 
Each edge $e\in E$ is associated with a nonnegative and nondecreasing latency function $\ell_e(\cdot)$, which takes $f_e$, the volume of flow using $e$, as the independent variable. Under $f$, the latency of a path $P$
is 
$\ell_P(f)=\sum_{e\in P}\ell_e(f_e)$, the sum of its edges' latencies, and the latency of a traveler equals the latency of his OD path along which he travels. 
Every traveler is \emph{selfish} in that he is only concerned about his
own latency. Given others' routing choices, he always chooses one of his OD paths with the lowest latency. The self-interested behaviors induce an \emph{information constrained routing game}, denoted by quadruple $(G,\ell,r,\mathcal E)$, where $\ell=\{\ell_e~|~e\in E\}$ denotes the set of latency functions, and $r=(r_1,\ldots,r_k)$  and $\mathcal E=(E_1,\ldots,E_k)$ describe the traffic rates and information sets of all $k$ types of travelers, respectively.  The Nash equilibrium of the game is often referred to as an \emph{Information-Constrained Wardrop Equilibrium} ({ICWE}).
\begin{definition}[ICWE flow] \cite{ibp18}
A flow $f$ is an \emph{ICWE flow} if for every type-$j$ traveler and every $o_{t(j)}$-$d_{t(j)}$ path $P\in \mathcal{P}_j$ with $f_P^j>0$, it follows that $\ell_P(f) \le \ell_{P'}(f)$ for each $P'\in \mathcal{P}_j$.    
\end{definition}
 The game always admits an ICWE flow, which is essentially unique: if $f$ and $\hat{f}$ are both ICWE flows, then $\ell_e(f_e)=\ell_e(\hat{f}_e)$ for every $e\in E$ \cite{ibp18}. By the definition and uniqueness, given ICWE flow $f$, for each $j\in[k]$, let  $\ell_j(f)=\ell_P(f)$ with $f_P^{(j)} >0$ denote the \emph{equilibrium latency} of type-$j$ travelers.

\subsection{Information Braess' Paradox}
Given game $(G,\ell,r,\mathcal E)$, 
we say $\tilde{\mathcal E}=\{\tilde{E}_1,\ldots,\tilde{E}_k\}$ is an information extension (of $\mathcal E$) if $E_1 \subset \tilde{E}_1$ and $E_j= \tilde{E}_j$ for $ j=2,\ldots,k$.
\begin{definition}[IBP]\cite{ibp18}
   Let $f$ and  $\tilde{f}$  be ICWE flows of  games $(G,\ell,r,\mathcal E)$ and $(G,\ell,r,\tilde{\mathcal E})$, respectively. If $\ell_1(\tilde{f})>\ell_1(f)$, then \emph{Information Braess' Paradox} (IBP) occurs, and quintuple $(G,\ell,r,\mathcal E,\tilde{\mathcal E})$ is called an IBP instance. 
\end{definition} 
The IBP instance encompasses two information constrained routing games $(G,\ell,r,\mathcal E)$ and $(G,\ell,r,\tilde{\mathcal E})$ with $f$ and $\tilde{f}$ being their ICWE flows respectively. The only difference between $\mathcal E$ and $\tilde{\mathcal E}$ is that $E_1 $ is properly contained in $ \tilde{E}_1$. By $\ell_1(\tilde{f})>\ell_1(f)$, we see that type-$1$ travelers have their information set expanded but suffer from a higher latency, while the information sets of other travelers remain unchanged. For a network $G$, if an IBP instance $(G,\ell,r,\mathcal E,\tilde{\mathcal E})$ exists, then we say $G$ is \emph{not immune to IBP}; otherwise, we call $G$ an \emph{IBP-free} network. 

For a network with OD pairs, we assume for convenience that each edge and each vertex of the network are contained in at least one OD path.  

\subsection{IBP-free Networks with a Single OD Pair} \label{sec:single}

Acemoglu et al.\ \cite{ibp18} established a sufficient and necessary condition for an undirected network with a single OD pair to be IBP-free (see Theorem~\ref{thm:sli}). To present their result, we need to introduce three classes of networks with a single OD pair $(o,d)$:  series-parallel networks, linearly independent networks, and series of linearly independent networks, where $\{o,d\}$ are referred to as their \emph{terminal sets}. These networks provide us with important structural properties for characterizing IBP-freeness in networks with one or more OD pairs. 

\begin{definition}[SP network]\label{def:sp}
A network with a single OD pair is called \emph{series-parallel} (SP) if any two OD paths never pass through an edge in opposite directions.
\end{definition}

We can use recursive operations to equivalently characterize SP networks (see, e.g., \cite{milchtaich2006network}). By connecting two networks with a single OD pair \emph{in series}, we mean joining the destination vertex of one network with the origin vertex of the other network (the origin of the first network becomes the origin of the new network, and the destination of the second network becomes the destination of the new network). By connecting two networks with a single OD pair \emph{in parallel}, we mean identifying the origin vertices of two networks to form the origin of the new network and identifying the destination vertices of two networks to form the new destination of the new network. 

\begin{figure}[ht!]
\centering
\begin{subfigure}[b]{0.3\textwidth}
\centering
\includegraphics[width=\textwidth, trim = 0cm 0cm 0cm 0cm, clip]{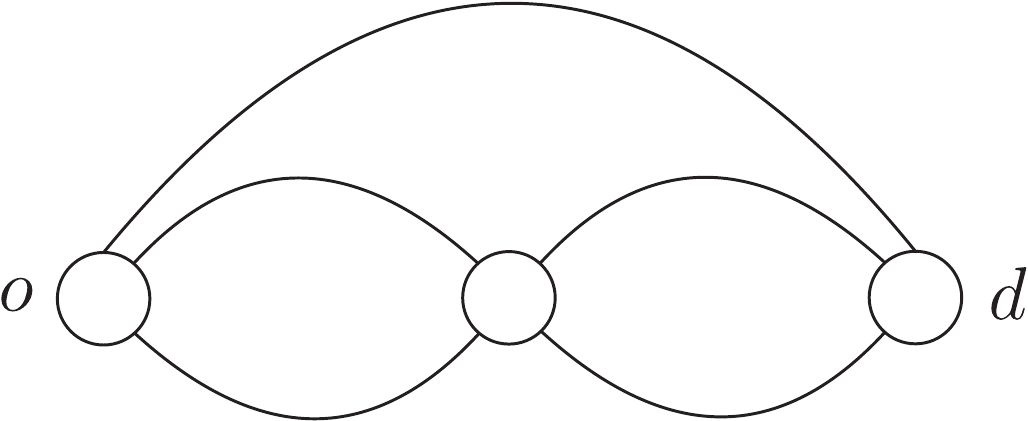}
\subcaption{\text{SP not SLI}}
\end{subfigure}
\hspace{0mm}
\begin{subfigure}[b]{0.3\textwidth}
\centering
\includegraphics[width=\textwidth, trim = 0cm 0cm 0cm 0cm, clip]{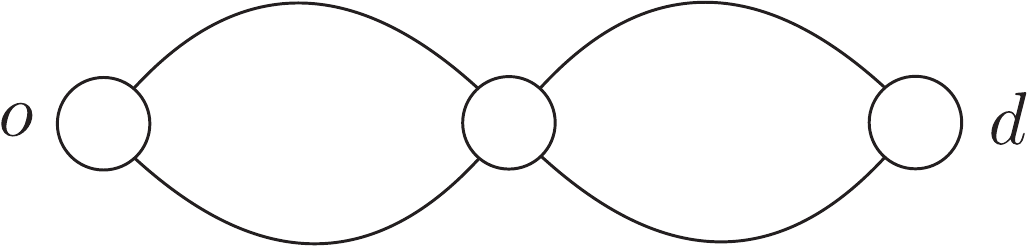}
\subcaption{\text{SLI not LI}}
\end{subfigure}
\hspace{0mm}
\begin{subfigure}[b]{0.3\textwidth}
\centering
\includegraphics[width=\textwidth, trim = 0cm 0cm 0cm 0cm, clip]{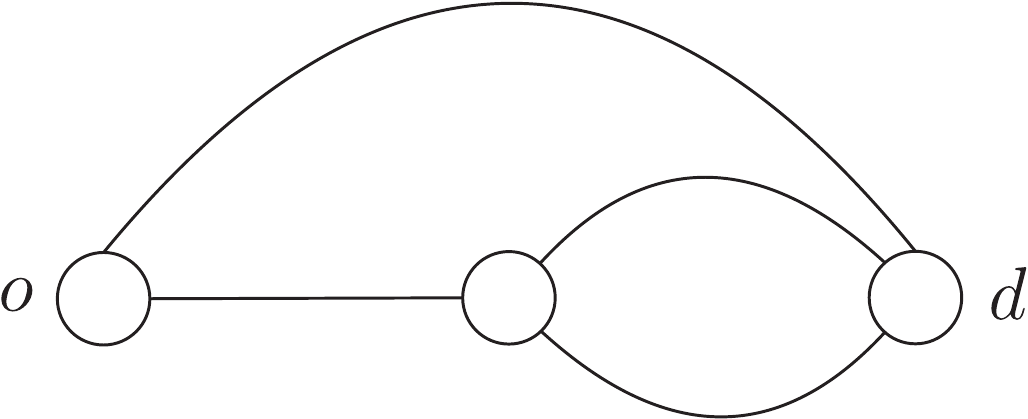}
\subcaption{\text{LI}}
\end{subfigure}
\caption{Three classes of  networks with a single OD pair}
\label{fig:sp_sli_li}
\end{figure}



A network with a single OD pair is SP  if and only if it is a single-edge network, or a connection of  two SP networks in series, or a connection of two SP networks in parallel.   
Another characteristic of SP networks is that the network in Figure~\ref{fig:thetagraph} can not be embedded in a SP network \cite{milchtaich2006network}. SP networks are immune to BP \cite{milchtaich2006network} but not immune to IBP \cite{ibp18}. An important subclass of SP networks consists of linearly independent ones, which have been shown to be IBP-free \cite{ibp18}.

\begin{definition}[LI networks]\label{def:li}
A network with a single OD pair is called \emph{linearly independent} (LI) if each OD path has an edge that does not belong to any other OD path.
\end{definition}

LI networks can also be characterized using recursive graphical operations \cite{holzman2003}. A network with a single OD pair is LI if and only if it is a single-edge network, or a connection of two LI networks in parallel, or a connection of a single edge and an LI network in series.
Another subclass of SP networks is obtained by connecting LI networks in series..

\begin{definition}[SLI network]\label{def:sli}
A network with a single OD pair is \emph{series of linearly independent} (SLI) if it is an LI network or it is a connection of two SLI networks in series.
\end{definition}
It is worth noting that the connection of two LI networks in series may not be an LI network (see Figure~\ref{fig:sp_sli_li}(b)), and the connection of two SLI networks in parallel may not be an SLI network (see Figure~\ref{fig:sp_sli_li}(a)). 

\begin{theorem}\cite{ibp18}\label{thm:sli}
A network with a single OD pair is IBP-free if and only if it is an SLI network.
\end{theorem}

A vertex in a connected graph is a \emph{cut vertex} if its removal from the graph leaves the graph unconnected. A \emph{block} of a graph is a maximal connected subgraph without any cut vertex. SP networks (and thus their special cases, SLI networks and LI networks)  possess a structure known as {\em block-chain}. It means that an SP network can be constructed by connecting several SP blocks in series, where the OD pair of each block is uniquely determined by the OD pair of the whole SP network~\cite{chen2016network}. 
In particular, \emph{an SLI network is a connection of several LI blocks in series}. See Figure~\ref{fig:ibp-sli-ex} for an illustration, where the OD pair of block $B_h$ is $(v_h,v_{h+1})$, $h\in[4]$. If a terminal (either the origin or the destination) of an LI block is not a terminal of the entire SLI network, then it is a cut vertex of the entire SLI network. For example, $v_2,v_3,v_4$ in Figure~\ref{fig:ibp-sli-ex} are those terminals.
\begin{figure}[ht!]
\centering
\includegraphics[width=0.7\textwidth]{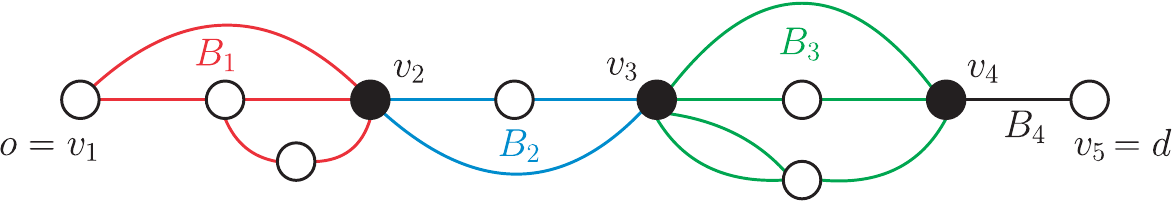}\\
\caption{An SLI network, which is the chain of blocks $B_1,B_2,B_3,B_4$.}
\label{fig:ibp-sli-ex}
\end{figure}

\subsection{IBP-free Networks with Multiple OD Pairs}\label{sec:multi}

In this paper, we completely characterize the structure of undirected networks with multiple OD pairs that are immune to IBP (see Theorem \ref{thm:main_0}), strengthening the sufficient condition by \cite{ibp18} (see Theorem~\ref{thm:coin}) and solving an open question proposed in \cite{ibp18}.

When network $G$ has $n\, (\ge2)$ OD pairs $(o_i,d_i)_{i=1}^n$,   we call $G$ a network with \emph{multiple} OD pairs.
For every $i\in[n]$, we reserve symbol $G_i$ to denote the subnetwork with terminal set $\{o_i,d_i\}$, or equivalently with a single OD pair $(o_i,d_i)$, which consists of all   $o_i$-$d_i$ paths in $G$. To characterize IBP-free networks, in view of Theorem~\ref{thm:sli}, it suffices to consider the scenario in which each  $G_i$ is an SLI network, and therefore its blocks are all LI. To investigate the mutual influence of routing in two subnetworks $G_i$ and $G_j$, the first step is to study the structural properties of their common part $G_i\cap G_j$. 

\begin{proposition}\cite{chen2016network}\label{pro:block}
Suppose that $G_i$ and $G_j$ are SLI subnetworks of $G$. Let $B_i$ and $B_j$ be blocks of $G_i$ and $G_j$, respectively. If $E(B_i)\cap E(B_j)\neq \emptyset$, then $B_i=B_j$ is a common block of $G_i$ and $G_j$. 
\end{proposition}

It is worth noting that a common block of $G_i$ and $G_j$ may have different terminal sets when restricted to the specific subnetworks. See Figure~\ref{fig:ibp-C3} for an illustration, where $G_1=G_2$ is a block with $\{o_1,d_1\}\ne\{o_2,d_2\}$.

\begin{definition}\cite{chen2016network}\label{def:coin}
Suppose that $G_i$ and $G_j$ are SLI. A block $B$ of $G$ is called {\em a coincident block} of $G_i$ and $G_j$ if it is a common block of $G_i$ and $G_j$, and the terminal set of $B$ in $G_i$ is the same as that in $G_j$. 
\end{definition}

In Definition~\ref{def:coin}, the orders of terminals (which is the origin and which is the
destination) of $B$ in $G_i$ and $G_j$ do not have to be the same. Combining the SLI condition for IBP-free networks with a single OD pair (see Theorem~\ref{thm:sli}) and the coincident condition for characterizing BP-freeness \cite{chen2016network}, Acemoglu et al.\ \cite{ibp18} established the following sufficient condition for excluding IBP in networks with multiple OD pairs.
\begin{theorem}\cite{ibp18}\label{thm:coin}
Let $G$ to be a network with $n$ OD pairs. For every $i\in [n]$, let $G_i$ be the subnetwork of $G$ with terminal set $\{o_i,d_i\}$. Then $G$ is IBP-free  \emph{if} $G$ satisfies the following conditions:
\begin{itemize}
\item SLI condition: for every $ i\in [n]$, $G_i$ is SLI.
\item Coincident condition: for any $i,j\in [n]$ with $i\neq j$, either $E(G_i)\cap E(G_j)=\emptyset$ or the graph induced by $E(G_i)\cap E(G_j)$ consists of all coincident blocks of $G_i$ and $G_j$. 
\end{itemize}

\end{theorem}
The coincident condition is not necessary, as shown by the cycle network $G$ with three vertices and two OD pairs depicted in Figure \ref{fig:ibp-C3}. Although $G=G_1=G_2$ is a non-coincident block, Acemoglu et al.\  \cite{ibp18} proved that $G$ is IBP-free.
 
\begin{figure}[ht!]
\centering
\includegraphics[width=0.26\textwidth]{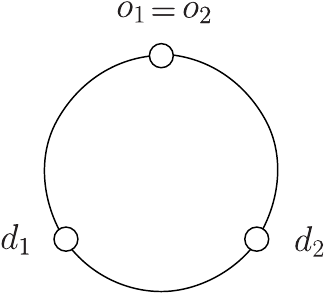}\\
\caption{A cycle network (non-coincident block) that is IBP-free}
\label{fig:ibp-C3}
\end{figure}

The next theorem represents our main contribution.
Building on the results of \cite{ibp18}, we establish the following necessary and sufficient condition, showing that cycle networks are \emph{the sole exception} not covered by the coincident condition.

\begin{theorem}[Main Result]\label{thm:main_0} 
Let $G$ to be a network with $n$ OD pairs. For every $i\in [n]$, let $G_i=(V_i,E_i)$ be the subnetwork with terminal set $\{o_i,d_i\}$. Then $G$ is IBP-free \emph{if and only if} $G$ satisfies the following conditions:
\begin{itemize}
\item SLI condition: for every $ i \in [n]$, $G_i$ is SLI.
\item Common block condition: for any $i,j\in [n]$ with $i\neq j$, either $E(G_i)\cap E(G_j)=\emptyset$ or the graph induced by $E(G_i)\cap E(G_j)$ consists of all common blocks of $G_i$ and $G_j$, where each common block is either a coincident block or a cycle.
\end{itemize}
\end{theorem}

Our technical contributions primarily lie in the following two aspects: (1) We identified IBP instances for non-coincident blocks that are not cycles. (2) We proved that cycle networks are IBP-free. The analysis for both aspects requires a thorough understanding of the interaction between network structures, selfish routing games, and the mechanisms of IBP. It involves disentangling and reducing the complex interweaving of intricate structures and routing dynamics. 

While  Roman and Turrini \cite{roman2019multi} asserted to have proven the IBP-freeness of cycle networks,  their proof contained irremediable flaws. Specifically, Proposition 3 in \cite{roman2019multi}, intended to mirror Lemma 1(b) of \cite{ibp18}, claimed 
that in a cycle network, after the information extension, there must exist a type of travelers whose latency does not increase. However, the authors \cite{roman2019multi} did not provide a complete proof for the general case. 
Furthermore, the proof of their main theorem (Theorem 3 in \cite{roman2019multi}) is flawed, as detailed in Appendix~\ref{apx:a}. The issues remain unresolved in Roman's thesis \cite{roman2021routing}. Given these shortcomings, we present a novel approach to prove the IBP-freeness of cycle networks. The proofs, which constitute the most technical part of our paper, are elaborated in Section~\ref{sec:cycle} and Appendix~\ref{app:cycle}. Our method is completely different from that of Roman and Turrini, offering a complete analysis of the problem.

\subsection{Related Work} 
The most related work for IBP-free networks, which has been discussed in Subsections \ref{sec:single} and~\ref{sec:multi}, includes the SLI characterization by \cite{ibp18} for the case of a single OD pair, and the sufficient conditions of \cite{ibp18,roman2019multi} for the case of multiple OD pairs.  

The study of the counterintuitive phenomenon in routing games was started by Braess in 1968~\cite{braess1968}. Milchtaich \cite{milchtaich2006network} introduced a series-parallel characterization for excluding BP in undirected networks with a single OD pair. Chen et al.\ \cite{chen2016network} extended this solution by providing a complete characterization for all undirected and directed networks with multiple OD pairs that are immune to BP, incorporating the series-parallel and coincident conditions.
BP appears not only in traffic networks, but also in communication networks \cite{orda1993competitive,korilis1997achieving,kelly1998rate,low1999optimization}, mechanical systems and electrical circuits \cite{cohen1991paradoxical,cohen1997congestion}. 
In addition, BP serves as the motivation for a network design problem: determining which edges should be removed from a given network to achieve the optimal flow at equilibrium. Roughgarden \cite{roughgarden2006severity} provided optimal inapproximability results and proposed approximation algorithms for tackling this network design problem. 

Traveler heterogeneity has garnered significant research attention. Milchtaich \cite{milchtaich1996congestion} studied the situation that different travelers may have different latency functions. Cole et al.\  \cite{cole2018does}  focused on the routing games where diversity emerges from the trade-off between two criteria and characterized the topology of networks for which diversity is never harmful. Meir and
Parkes \cite{meir2015congestion} proposed a routing game model in which travelers face uncertainty regarding the paths chosen by others. Sekar et al.\ \cite{sekar2019uncertainty}  considered the selfish routing game in which travelers may overestimate or underestimate their traffic latencies. Wu et al. \cite{wu2021value} studied the equilibrium route choices and traffic latencies within a heterogeneous information system, considering an uncertain state that impacts the latencies of edges in the network.


\paragraph{\bf Organization.} The rest of the paper is organized as follows:  Section~\ref{sec:pre} introduces additional notations and presents preliminaries. Sections \ref{sec:2net} is dedicated to proving the ``only if'' part of our characterization for IBP-free networks, which asserts that under the SLI condition, non-coincident blocks are either cycles or admit IBP.  Section \ref{sec:cycle} establishes the ``if'' part of our characterization,  showing that all cycle networks with multiple OD pairs are IBP-free.   
Section \ref{sec:conclude} concludes the paper with remarks on possible directions of future research. Due to space limitations, some proofs are omitted from the main body and provided in the Appendix.

\section{Preliminaries}\label{sec:pre}

Given an information constrained routing game $(G,\ell,r,\mathcal E)$ that is IBP-free, we can decompose the game into several local games based on the block-chain structure of $G$. By Theorem~\ref{thm:sli}, we assume that network $G$ satisfies the SLI condition in Theorem~\ref{thm:main_0}.  Thus each subnetwork $G_i$ processes a block-chain structure (recall Figure~\ref{fig:ibp-sli-ex}). Furthermore, using Proposition~\ref{pro:block}, we can decompose the network $G$ into a number of LI blocks $B_1,B_2,\ldots,B_l$, each of which is a common block of several $G_i$'s. For ease of presentation, if $B_s \subseteq G_i$, i.e., $B_s$ is a block of $G_i$, we denote the OD pair of $B_s$ in the block-chain of $G_i$ as $(o_i, d_i)$. In this way, we may consider a local game $(B_s,\ell_{B_s},r_{B_s},\mathcal E_{B_s})$ restricted to $B_s$, where $\ell_{B_s}$ is the restriction of $\ell$ to $B_s$, $\mathcal  E_{B_s}$ is the restriction of $\mathcal E$ to $B_s$, and the traffic rates are \emph{consistent} with those in $G$, i.e., 
\begin{center}
    $r_{B_s}(j)=r_j$ if $B_s \subseteq G_{t(j)}$, and  $r_{B_s}(j)=0$ otherwise.
\end{center} Note that although the real participants of the game are only travelers whose OD paths pass some edges in $B_s$, i.e., type-$j$ travelers with $ G_{t(j)}\supseteq B_s$, for notational convenience, we consider the other travelers (if any) as dummy participants of the game with traffic rate 0 for which we do not need to specify OD pairs or information sets.  Let $f_s$ be an ICWE flow of $(B_s,\ell_{B_s},r_{B_s},\mathcal E_{B_s})$ and $\ell_j(f_s)$ be the ICWE latency of type-$j$ travelers in network $B_s$; if $r_{B_s}(j)=0$, then we set $\ell_j(f_s)=0$. Let $f$ be an ICWE flow of $(G,\ell,r,\mathcal E)$. It is easy to see that the equilibrium latency of a traveler in $G$ is the sum of his equilibrium latency in all local games.
\begin{proposition}\label{pro:series}
For every $j\in [k]$, $\ell_j(f)=\sum_{s=1}^{l}\ell_j(f_s)$.
\end{proposition}
\begin{proof}
Let $f_{B_s}$ denote the restriction of $f$ on $B_s$. For every $B_s\subseteq G_{t(j)}$, it is clear that $f_{B_s}$ is an ICWE flow of $(B_s,\ell_{B_s},r_{B_s},\mathcal E_{B_s})$ and  $\ell_j(f_{B_s})=\ell_j(f_s)$ by the uniqueness of ICWE. Because $G_{t(j)}$ is the connection of its blocks in series, we have $\ell_j(f)=\sum_{B_s\subseteq G_{t(j)}}\ell_j(f_{B_s})=\sum_{B_s\subseteq G_{t(j)}}\ell_j(f_s)=\sum_{s=1}^{l}\ell_j(f_s)$, where the last equality is implied by the setting that
 $\ell_j(f_s)=0$ whenever $B_s\nsubseteq G_{t(j)}$.
\end{proof}

\begin{lemma}\label{lem:subgame}
Let $G$ be a network with multiple OD pairs that satisfies the SLI condition. Let  $B_1,B_2,\ldots,B_l$ be the blocks of $G$, each of which is an  LI network with the OD pairs determined by $G$ and $G$'s OD pairs. Then $G$ is an IBP-free network if and only if  $B_s$ is an IBP-free network for every $s\in[l]$.
\end{lemma}
\begin{proof} To see the sufficiency, suppose that $(G,\ell,r,\mathcal E,\tilde{\mathcal E})$ is an IBP instance. Let $f$ be an ICWE flow of $(G,\ell,r,\mathcal E)$ and $\tilde{f}$ be an ICWE flow of $(G,\ell,r,\tilde{\mathcal E})$. By Proposition~\ref{pro:series}, we have $\ell_1(f)=\sum_{s=1}^{l}\ell_1(f_s)$ and $\ell_1(\tilde{f})=\sum_{s=1}^{l}\ell_1(\tilde{f_s})$. Since $\ell_1(f)<\ell_1(\tilde{f})$, there must be an $s\in [l]$ such that $\ell_1(f_s)<\ell_1(\tilde{f_s})$. It follows that $(B_s,\ell_{B_s},r_{B_s},\mathcal E_{B_s},\tilde{\mathcal E}_{B_s})$ is an IBP instance. 

To see the necessity, suppose that  $(B_{s^*},\ell_{B_{s^*}},r_{B_{s^*}},$ $\mathcal E_{B_{s^*}},\tilde{\mathcal E}_{B_{s^*}})$ with $\ell_1(f_{s^*})<\ell_1(\tilde{f}_{s^*})$ is an IBP instance for some ${s^*}\in[l]$. We construct an IBP instance $(G,\ell,r,\mathcal E,\tilde{\mathcal E})$ as follows. Let $\ell_e\equiv 0$ for every edge $e\in E(G)-E( B_{s^*})$. For type-$j$ travelers in $G$, let their traffic rates be consistent with $r_{B_{s^*}}$, let their information set before (resp.\ after) information expansion when restricted to $B_{s^*}$ be as in $\mathcal E_{B_{s^*}}$ (resp.\ $\tilde{\mathcal E}_{B_{s^*}}$), and let their information sets   (before and after information expansion) when restricted to $B_u$ (with $u\ne s^*$ and $B_u\subseteq G_{t(j)}$) be $E(B_u)$. Let $f$ and $\tilde f$ be ICWE flows of $(G,\ell,r,\mathcal E)$ and  $(G,\ell,r,\tilde{\mathcal E})$, respectively. By Proposition~\ref{pro:series}, $\ell_1(f)=\sum_{s=1}^{l}\ell_1(f_s)=\ell_1(f_{s^*})$ and $\ell_1(\tilde{f})=\sum_{s=1}^{l}\ell_1(\tilde{f_s})=\ell_1(\tilde{f}_{s^*})$.  Since $\ell_1(f_{s^*})<\ell_1(\tilde{f}_{s^*})$, we have $\ell_1(f)<\ell_1(\tilde{f})$, showing that $(G,\ell,r,\mathcal E,\tilde{\mathcal E})$ is an IBP instance. 
\end{proof}
 
For an edge $e \in E(G)$, we use $G/e$ to denote the network obtained from $G$ by contracting $e$ into a new vertex $v_e$. If \emph{exactly one} of the two vertices incident to $e$ is a terminal vertex of $G$, then the corresponding terminal of $G/e$ is the new vertex $v_e$. We use $G\setminus e$ to denote the network obtained from $G$ by removing $e$. 


\begin{definition} [Nice embedding]
A network $H$ is nicely embedded in the network $G$ if $H$ can be obtained from $G$ by recursive edge removals and contractions under the condition that no terminal set of the networks in the process is merged by the contractions.
\end{definition}


\begin{lemma}\label{lem:oper}
If network $H$ can be nicely embedded in the network $G$ and $H$ is not immune to IBP, then $G$ is not immune to IBP.
\end{lemma}
\begin{proof}
We only need to consider the situation that $H$ is $G\setminus e$ or $G/e$ for some edge $e\in E(G)$. 
We use $\ell$ to denote the latency function set in $H$ for convenience. 
If  $(H,\ell,r,\mathcal E,\tilde{\mathcal E})$ is an IBP instance, then we set $\mathcal E'=\mathcal E$ and $\tilde{\mathcal E}'=\tilde{\mathcal E}$ when $H=G\setminus e$, and set   $\mathcal E'=\{E\cup\{e\}~|~E\in\mathcal E\}$ and $\tilde{\mathcal E}'=\{E\cup\{e\}~|~E\in\tilde{\mathcal E}\}$ when $H=G/e$.
Setting $e$'s latency function by $\ell_e\equiv 0$, it is easy to see that $(G,\ell,r,\mathcal E',\tilde{\mathcal E}')$ is an IBP instance. 
\end{proof}

\section{Necessity Proof}\label{sec:2net}
In this section, we establish the necessity of the conditions stated in Theorem~\ref{thm:main_0}. It is evident from Theorem~\ref{thm:sli} that the SLI condition is a prerequisite for IBP-freeness.  It remains to prove the necessity of the common block condition.  To accomplish this,  we only need to consider an \emph{IBP-free network} $G$ with two OD pairs such that $G_1$ and $G_2$ are SLI. Furthermore, by Proposition~\ref{pro:block} and Lemma~\ref{lem:subgame}, we may assume that $G=G_1=G_2$ is the non-coincident common block of $G_1$ and $G_2$. 
Our objective is to show that \emph{$G$ is a cycle}, thereby confirming the necessity of the condition.

\subsection{An IBP Instance}
In this subsection, we show that an IBP can occur in the following network $F_1$ with three vertices, two OD pairs $(o_i,d_i)_{i=1}^2$, and latency functions $\ell$, as depicted in Figure~\ref{fig:2odibp}(a). Specifically, $V(F_1)=\{u,v,w\}=\{o_1,d_1,o_2,d_2\}$, $E(F_1)=\{e_1,e_2,e_3,e_4\}$, and latency functions are written beside the corresponding edges, e.g., $\ell_{e_1}(x)=0$, $\ell_{e_2}(x)=4x$, $\ell_{e_3}(x)=x+22$ and $\ell_{e_4}(x)=10+2x$. 

\captionsetup[subfigure]{labelformat=empty}
\begin{figure}[h]
\centering
\begin{subfigure}[b]{0.3\textwidth}
\centering
\includegraphics[width=\textwidth, trim = 0cm 0cm 0cm 0cm, clip]{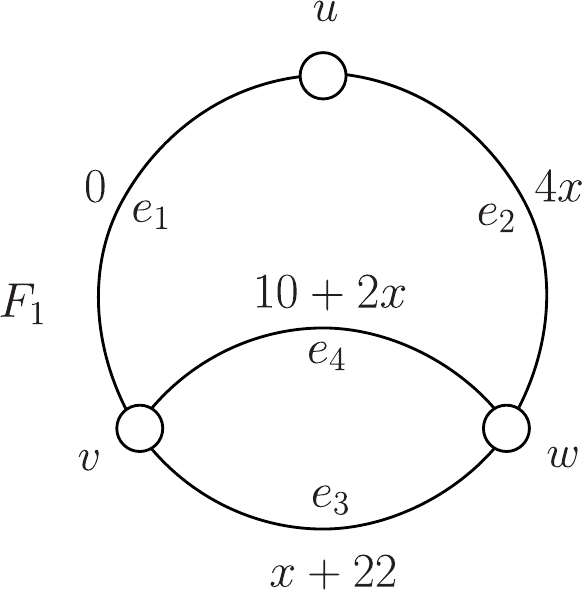}
\subcaption{\text{~(a)~network $(F_1,\ell)$}\\ ~}
\end{subfigure}
\hspace{1mm}
\begin{subfigure}[b]{0.27\textwidth}
\centering
\includegraphics[width=\textwidth, trim = 0cm 0cm 0cm 0cm, clip]{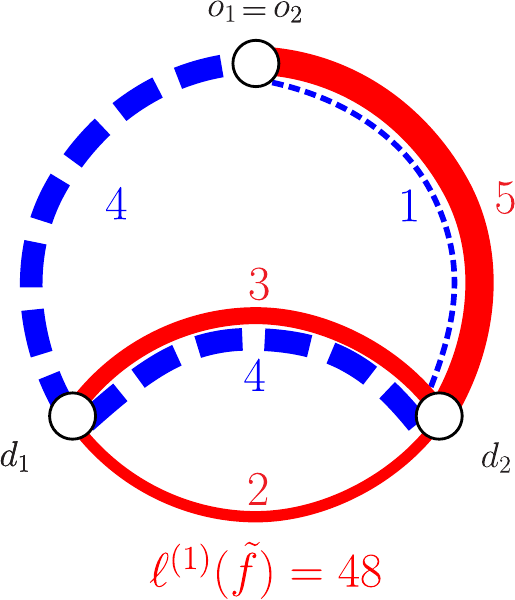}
\subcaption{\text{~\hspace{2mm}~(b)~ICWE flow $\tilde{f}$}\\ \text{$~\hspace{6.4mm}~$ with $o_2=o_1$}}
\end{subfigure}
\hspace{1mm}
\begin{subfigure}[b]{0.3\textwidth}
\centering
\includegraphics[width=\textwidth, trim = 0cm 0cm 0cm 0cm, clip]{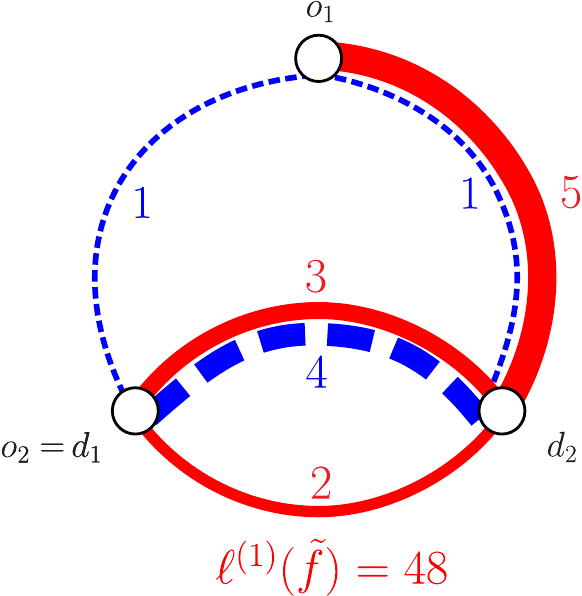}
\subcaption{\text{~\hspace{6mm}~(c)~ICWE flow $\tilde{f}$}\\ \text{$~\hspace{10mm}~$ with $o_2=d_1$}}
\end{subfigure}
\caption{A network with two OD pairs that admits IBP}
\label{fig:2odibp}
\end{figure}

Suppose that $k=2$ and that for $j=1,2$, type-$j$ travelers have OD pair $(o_j,d_j)$ and traffic rate $r_j=5$.  By the symmetry between OD pair indices and that between the vertices in the same OD pair, we may assume $(o_1,d_1)=(u,v)$, $d_2=w$ and $o_2\in\{o_1,d_1\}$. The information set  of type-1 travelers changes from $E_1=\{e_2,e_3\}$ to larger $\tilde E_1=\{e_2,e_3,e_4\}$, while the information set of type-2 travelers is always $E_2=\{e_1,e_2,e_4\}$. Let $f$ and $\tilde f$ denote the ICWE flows before and after type-1 travelers enlarge their information set, respectively.
\begin{itemize}
\item Clearly, $f^1_{e_2e_3}=5$,   $f^2_{e_1e_4} =5$ when $o_2=o_1$ and $f^2_{e_4} =5$ when $o_2=d_1$, from which we deduce that under $f$ type-1 travelers experience latency $\ell^1(f)=47$.
\item After type-1 travelers known about $e_4$, some of them will use this ``new'' edge instead of $e_3$. Easy computation shows that these travelers amount to $\tilde f^1_{e_4}=3$, and the remaining  $\tilde f^1_{e_3}=2$ amount of type-1 travelers stick to $e_3$; accordingly, type-2 travelers will split their routes:  when $o_2=o_1$ (resp.\ $o_2=d_1$), amount 4 sticking to  their left path $e_1e_4$ (resp.\ lower path $e_4$) and amount $1$  changing to their right path $e_2$ (resp.\ upper path $e_1e_2$). See Figure~\ref{fig:2odibp}(b) and (c) for an illustration. It is easy to check that $\tilde f$ is indeed an ICWE flow under which type-1 travelers experience a higher latency $\ell^1(\tilde f)=48$.
\end{itemize}
In conclusion, $(F_1,\ell,r,\{E_1,E_2\},\{\tilde E_1,E_2\})$ constitutes an IBP instance, saying that $F_1$ is not immune to IBP. Recalling Lemma~\ref{lem:oper} and the IBP-freeness of $G$, we have the following corollary.
\begin{corollary}\label{cor:f1}
    $F_1$ cannot be nicely embedded in $G$.
\end{corollary}

\subsection{Graphical Structures}

Since $G=G_1=G_2$ is a non-coincident block of $G_1$ and $G_2$, we have $\{o_1,d_1\}\ne\{o_2,d_2\}$. A connected graph is {\em $2$-connected} if it contains more than $2$ vertices and has no cut vertices. Since $G$ is a non-coincident block, it contains at least $3$ vertices, and therefore is 2-connected. The reader is referred to \cite{diestel2018graph} for properties of 2-connected graphs. Setting $S=\{o_1,d_1,o_2,d_2\}$, we have $|S|\in\{3,4\}$ because $G$ is non-coincident. 
Note that if $F_1$ can be nicely embedded in $G$, then there is a mapping $b:F_1 \rightarrow G$ such that $b^{-1}(S)=\{u,v,w\}$ and $|b^{-1}(\{o_i,d_i\})|=2$ for $i=1,2$.
  
Recall that $G_1$ and $G_2$ are assumed to be SLI networks. So  the theta graph $F_2$ in Figure~\ref{fig:thetagraph} can neither be embedded in $G_1$ nor in $G_2$ \cite{milchtaich2006network}.
 
\begin{figure}[ht!]
\centering
\includegraphics[width=0.6\textwidth]{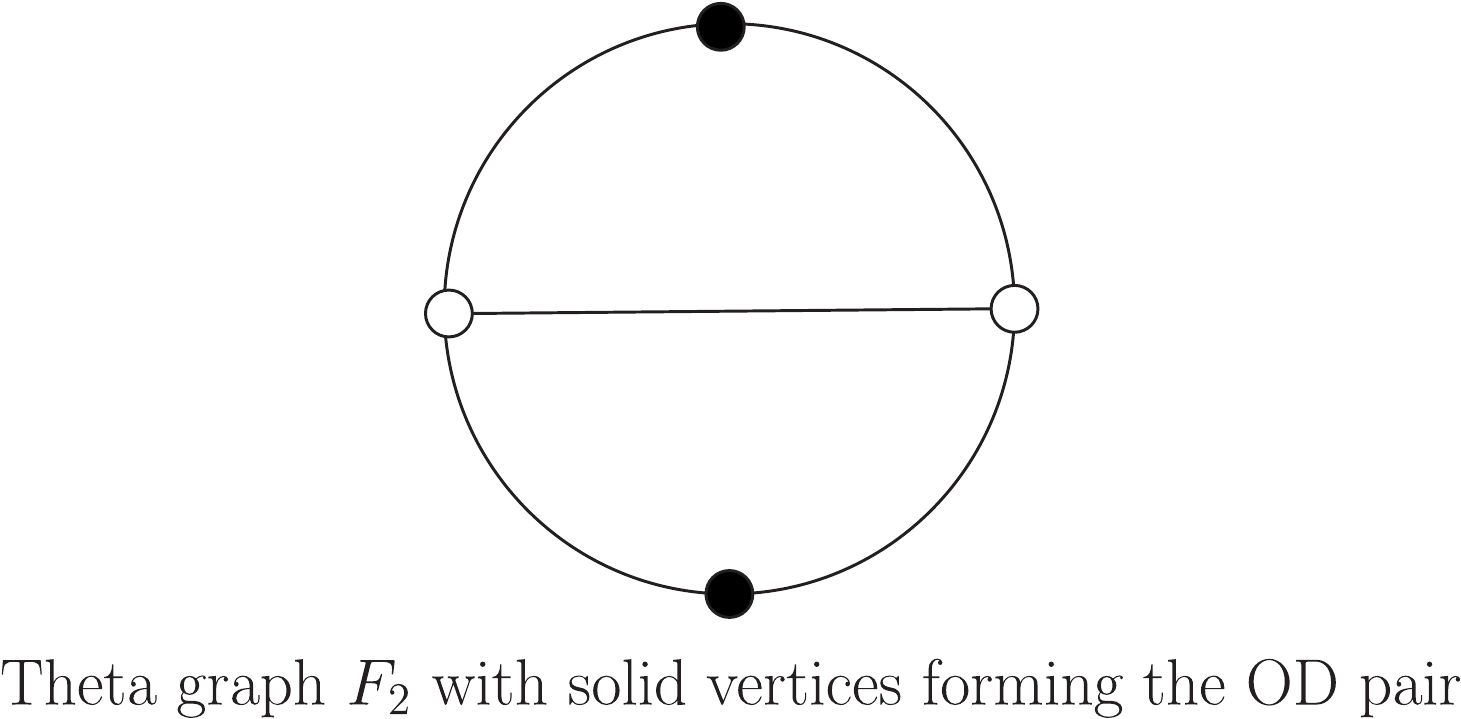}\\
\caption{A network that can not be embedded in LI networks}
\label{fig:thetagraph}
\end{figure}
 
\begin{lemma}\label{lem:3vertices}
If graph $G$ is not a cycle and contains a cycle $C$ with $|S-V(C)|\le1$, then $F_1$ can be {nicely embedded} in $G$.
\end{lemma}
\begin{proof} Suppose without loss of generality that $o_1,d_1,o_2\in V(C)$. Since $G$ is not a cycle, we may take $e\in E(G)-E(C)$ such that $d_2$ is an end of $e$ whenever $d_2\not\in V(C)$. By the 2-connectivity of $G$, there is a path $P$ in $G$ between distinct vertices $x,y$ on $C$ that goes through $e$ and is internally disjoint from $C$.  Recall that the theta graph $F_2$ (see Figure~\ref{fig:thetagraph}) is not embeddable in $G_1$. So one of the two $o_1$-$d_1$ paths in $C$ contains both $x$ and $y$. Now $C\cup P$ contains both OD pairs, i.e., $S$. Let $H$ be obtained from $C\cup P$ by iteratively contracting edges as many as possible under the condition that no two distinct vertices from $S$ are merged and that $x$ and $y$ are not merged.  If an edge with one end $z\in S\cup\{x,y\}$ is contracted, the vertex resulting from the contraction is named as $z$. In particular, when both ends of the contracted edge are in $S\cup\{x,y\}$, it must be the case that one end, say $z_1$, is in $S-\{x,y\}$ and the other, say $z_2$, is in $ \{x,y\}-S$, for which we consider the two ends merge to a new vertex $z_1=z_2\in S$. The contraction process contracts $C$ to a cycle $D$ with $V(D)=(S\cap V(C))\cup\{x,y\}$ and contracts $P$ to an $x$-$y$ path with $V(Q)=(S-V(C))\cup\{x,y\}$. Thus $H$ with $V(H)=S\cup\{x,y\}$ is edge disjoint union of $D$ and $Q$. To prove the lemma, it suffices  
to check that $F_1$ can be {nicely embedded} in $H$. In the remainder of the proof, we focus on vertices of $H$ instead of those in the original graph $C\cup P$.
	
By the construction of $H$, each edge in $H$ has both ends in $S\cup\{x,y\}$. If one of $x$ and $y$, say $x$, is outside $S$, then there is an edge $e\in E(H)$ that joins $x$ and a vertex $z\in S-\{y\}$, and $e$ should have been contracted in the above process of constructing $H$. So  we have $\{x,y\}\subseteq S$ and hence $V(H)=S$. Note that $x$ and $y$ have degree 3, and other vertices have degree 2 in $H$.

When $|S|=3$, we deduce that $H$ with $|V(H)|=3$ is $F_1$ itself. It remains to consider the case where $|S|=4$. If there is $e\in E(H)$  between $\{o_1,d_1\}$ and $\{o_2,d_2\}$ and $e$ does not join $x$ and $y$, then $H/e=F_1$ shows that $F_1$ can be {nicely embedded} in $H$. Next, we prove that such an edge $e$ does exist. When $x$ and $y$ have a common neighbor, at least one of the two edges incident with the common neighbor qualifies to be $e$. When $x$ and $y$ have no common neighbor, $x$ and $y$ are neighbors in both $D$ and $Q$, and there is a 3-edge  $x$-$y$ path in $D$ that contains $S$; it is easy to see that one of the three edges in the path qualifies to be the $e$ as desired.
\end{proof}

\begin{lemma} \label{lem:cycle}
Either $G$ is a cycle or $F_1$ can be {nicely embedded} in $G$.
\end{lemma}
\begin{proof}Suppose that $G$ is a 2-connected graph that is not a cycle. Let $C$ be a cycle in $G$ going through $o_1$ and $d_1$. By Lemma~\ref{lem:3vertices}, we only need to consider the case of $o_2\notin V(C)$. Similar to the above proof, the 2-connectivity of $G$ gives a path  $P$ between  distinct vertices $x,y$ on $C$ that goes through $o_2$ and is internally disjoint from $C$. In turn, excluding theta graph $F_2$ implies that both $x$ and $y$ are contained in one of the two $o_1$-$d_1$ paths in $C$; we denote the path by $Q$. Now $(C\setminus V(Q(x,y)))\cup P$ is a cycle in $G$ that contains $o_1,d_1,o_2$. It follows from Lemma \ref{lem:3vertices} that $F_1$ can be {nicely embedded} in $G$.
\end{proof}
The combination of Lemma~\ref{lem:cycle} and Corollary~\ref{cor:f1} enforces that $G$ is a cycle as desired. The ``only if'' part of Theorem~\ref{thm:main_0} has been verified.

\section{Sufficiency Proof}\label{sec:cycle}

In this section, we justify the ``if'' part of  Theorem~\ref{thm:main_0}. By contradiction, suppose that network $G$ with $n$ OD pairs satisfies the SLI condition and common block condition, but $G$ is not immune to IBP. It follows from the SLI condition and Proposition~\ref{pro:block} that every block of $G$ is a block of some $G_i$, $i\in[n]$, which is therefore LI. Since $G$ is not immune to IBP, by Lemma~\ref{lem:subgame}, some block of $G$, denoted by $C$, is not immune to IBP. In turn, Theorem~\ref{thm:sli} implies that $C$ is a network with $n\,(\ge2)$ OD pairs. Then, we deduce from Theorem \ref{thm:coin} that $C$ is not a coincident block of its subnetworks, which, along with the common block condition, enforces that $C$ is a cycle. To reach a contraction and prove the sufficiency, it suffices to establish the following IBP-freeness.

\begin{theorem}\label{thm:cycle_main}
A cycle $C$ with $n~(\ge2)$ OD pairs is IBP-free.
\end{theorem}

Due to space limitations, the proof of Theorem \ref{thm:cycle_main} is deferred to Appendix \ref{app:cycle}. The proof strategies can be summarized as follows:
\begin{itemize}
\item Reduction to fewer traveler types:
By Lemmas \ref{lem:type_n+1} and \ref{lem:type_n}, we reduce the problem to be the one with fewer types of travelers: If IBP happens, then it happens on a cycle network such that each OD pair is associated with \emph{only one} type of travelers. According to Lemma~\ref{lem:cir_change_path}, we can assume that the travelers of the same type choose the same one OD path before the information expansion and a \emph{different} OD path after the information expansion.
\item Increasing latencies and non-dominant traffic rates:
If IBP occurs, for every type of travelers, the latency of the path chosen before the information expansion will \emph{strictly increase} after the information expansion (see Lemma \ref{lem:key}). In addition, when IBP occurs, the traffic rate of any type of travelers is \emph{strictly less} than the total traffic rate of all the others (see Lemma \ref{lem:r_i}). Then, we directly derive that a cycle with two OD pairs is IBP-free in Lemma \ref{lem:2od},  significantly streamlining the proof process in \cite{ibp18}.
\item Reduction to fewer OD pairs via path coverage: 
In Lemma \ref{lem:union_cycle}, if there are two types of travelers such that the union of their chosen paths covers the entire cycle in an IBP instance, then we can construct a new IBP instance with \emph{fewer} OD pairs. By using this lemma and case analysis, it follows that a cycle with three OD pairs is IBP-free (see Lemma \ref{lem:3od}).
\item Reduction to fewer OD pairs via specific path choices:
According to Lemmas \ref{lem:4equal} -- \ref{lem:shrink_3}, if two types of travelers with specific path choices exist in an IBP instance, then we can construct a new IBP instance with \emph{fewer} OD pairs. Lemma \ref{lem:ge_le} states the relationship between the path choices and the amount of flow in edges. 
\item Symmetric structures: 
Next, we consider two structures on a cycle with $n$ OD pairs: a \emph{completely symmetric} distribution of terminal vertices and an \emph{almost symmetric} distribution of terminal vertices. These two structures are proven to be immune to IBP  in Lemmas \ref{lem:symmetric} and \ref{lem:almost_sym}.
\item Traffic rate and structure transformations:
In Lemma \ref{lem:r_to_int}, we convert an IBP instance in a cycle with $n$ OD pairs to one with all \emph{integer} traffic rates, and in Lemma \ref{lem:r_to_1}, we transform an IBP instance in a cycle with a finite number of OD pairs such that all traffic rates are equal to $1$. Moreover, Lemma \ref{lem:gen_to_sym} uses swap operations to convert a general  IBP instance (with $n$ OD pairs) into an instance on a completely symmetric structured cycle or an almost \emph{symmetric structured cycle}.
\item IBP-freeness:
Finally, by combining Lemmas \ref{lem:symmetric} -- \ref{lem:gen_to_sym}, we arrive at the main result of this section (Theorem \ref{thm:cycle_main}): a cycle with $n$ OD pairs is IBP-free.
\end{itemize}

\section{Concluding Remark}\label{sec:conclude}
In this paper, we have characterized all \emph{undirected} IBP-free networks with multiple OD pairs. Investigating the \emph{directed} counterpart presents an intriguing and potentially challenging avenue for future research. 

Our analysis could potentially be extended to other paradoxical settings. For example, classical BP represents the situation where adding an extra link to a network can sometimes strictly increase the travel times for travelers between one OD pair without decreasing travel times for all others. We propose exploring a ``strong increase'' version of BP, where adding an extra link to a network can sometimes strictly increase the travel times for all travelers whose OD paths contain the link. The analytical methods and characterization properties presented in this paper are also adaptable for studying networks that are immune to this stronger form of BP.

\bibliographystyle{unsrt}
\bibliography{routing}

\appendix

\section{Two Critical Flaws in the Proof of Theorem 3 in \cite{roman2019multi}}\label{apx:a}
This section raises valid concerns about the robustness of the proof of Theorem 3 in \cite{roman2019multi} by identifying the following two primary issues.
\begin{enumerate}
\item 
Inequality Derivation (Page 569, left column, fifth paragraph): The inequality $\sum_{e \in s_a} c_e(f_e(\tilde{x}))$ $\le \sum_{e \in s_a} c_e(f_e(x))$ is derived using the monotonicity of the functions $c_e$. This derivation requires that $\ f_e(\tilde{x}) \le f_e(x)$ holds for all $e \in s_a$. However, unfortunately, the single inequality $\sum_{e \in s_a} f_e(\tilde{x}) \le \sum_{e \in s_a} f_e(x)$ provided by the authors   is insufficient to guarantee the required monotonicity condition for each individual term. 
\medskip
\item 
Generality Oversight (Page 570, left column, first paragraph): The assertion ``without loss of generality'' fails to account for a crucial scenario: when the demand for $s_{\alpha i}$ decreases and belongs to $S_\alpha$, while the demand for $s_{\beta i}$ increases and belongs to $S_\beta$. This oversight is significant because this particular scenario cannot be used to support the result of Claim 3 in \cite{roman2019multi}. In fact, considering this case would lead to the opposite inequality, contradicting the intended conclusion.
\end{enumerate}

\section{The Proof of Theorem \ref{thm:cycle_main}}\label{app:cycle}

This section is devoted to the proof of Theorem~\ref{thm:cycle_main}.
Without loss of generality, assume that any two terminal vertices $x$ and $y$ are distinct. (If $x$ and $y$
  coincide at a single vertex, this vertex can be split into two different vertices, each incident with different (one of the) edges originally incident with the coincident vertex. Name these two vertices  $x$ and $y$, respectively, and add a zero-latency edge between them.) Additionally, it can also be assumed that all vertices on the cycle are terminal vertices in OD pairs. Otherwise, if there is a path of length at least $2$ between $x$ and $y$ that does not pass through any other terminal vertices, then the $x$-$y$ path  can be replaced with a single edge between $x$ and $y$. The latency function of this edge would be the summation of the latency functions of the edges along the $x$-$y$ path. In summary, we can assume that there are $2n$ vertices on the cycle $C$, corresponding to $n$ origin vertices and $n$ destination vertices.

There are exactly two paths between each $(o_i,d_i)$ pair on the cycle $C$, denoted as $\alpha_i$ and $\beta_i$, for which we have $C = \alpha_i \cup \beta_i$. Here, $C$ represents the set of edges on the cycle, and $\alpha_i$ and $\beta_i$ represent the sets of edges on the two paths, respectively.

For each OD pair $(o_i,d_i)$, there may be more than one type of travelers taking it as their OD pair. For ease of presentation, we use ``type-$(i,j)$ travelers'' to mean the $j$-th type of travelers whose OD pair is $(o_i,d_i)$. 

Let us assume that the type-$(1,0)$ travelers with OD pair $(o_1,d_1)$ have the expanded information sets. Let the path set before the information expansion be $\mathcal P_{(1,0)}=\{\alpha_1\}$, and after the information expansion, the path set becomes $\tilde{\mathcal P}_{(1,0)}=\{\alpha_1,\beta_1\}$. The remaining types of travelers with the unchanged information sets can be defined as follows: For each $i\in [n]$, 
\begin{itemize}
\item
The path set of type-$(i,1)$ travelers is $\mathcal P_{(i,1)}=\{\alpha_i,\beta_i\}$.
\item
The path set of type-$(i,2)$ travelers is $\mathcal P_{(i,2)}=\{\alpha_i\}$.
\item
The path set of type-$(i,3)$ travelers is $\mathcal P_{(i,3)}=\{\beta_i\}$.
\end{itemize}

In the game, there are $3n+1$ types of travelers. The following lemma shows that, aside from the type-$(1,0)$ travelers, it can be assumed that there are only type-$(i,1)$ travelers with the complete information for each $i\in [n]$.

\begin{lemma}\label{lem:type_n+1}
Suppose that $C$ is a cycle with $n$ OD pairs. If there is an IBP instance $(C,\ell,r,\mathcal E,\tilde{\mathcal E})$ with $3n+1$ types of travelers $\bigcup_{i=1}^n\{(i,1),(i,2),(i,3)\}\cup\{(1,0)\}$, then there is an IBP instance $(C,\ell',r',\mathcal E',\tilde{\mathcal E'})$ with $n+1$ types of travelers $\bigcup_{i=1}^n\{(i,1)\}\cup\{(1,0)\}$. 
\end{lemma}

\begin{proof}
For any $i\in[n]$, by the symmetry of $\alpha_i$ and $\beta_i$, we only show a new IBP instance without type-$(i,2)$ travelers. Let $r_{(i,2)}$ be the traffic rate of type-$(i,2)$ travelers. Since type-$(i,2)$ travelers always choose to take $\alpha_i$ before and after the information expansion of type-$(1,0)$ travelers, the flow of this part can be fixed on the latency function of the edges in path $\alpha_i$. That is, $\ell'_e(x)=\ell_e(x+r_{(i,2)})$ for each $e\in\alpha_i$ and $\ell'_e(x)=\ell(x)$ for each $e\in C\setminus\alpha_i$. Therefore, we can construct a new IBP instance without type-$(i,2)$ and type-$(i,3)$ travelers for each $i\in[n]$.
\end{proof}


Furthermore, we can construct an IBP instance that does not include type-$(1,1)$ travelers.

\begin{lemma}\label{lem:type_n}
Suppose that $C$ is a cycle with $n$ OD pairs. If there is an IBP instance $(C,\ell,r,\mathcal E,\tilde{\mathcal E})$ with $n+1$ types of travelers $\bigcup_{i=1}^n\{(i,1)\}\cup\{(1,0)\}$, then there is an IBP instance $(C,\ell',r',\mathcal E',\tilde{\mathcal E'})$ with $n$ types of travelers $\bigcup_{i=2}^n\{(i,1)\}\cup\{(1,0)\}$.
\end{lemma}
\begin{proof}
Denote the path sets for type-$(1,0)$ travelers before and after the information expansion as $\mathcal P_{(1,0)}=\{\alpha_1\}$ and $\tilde{\mathcal P}_{(1,0)}=\{\alpha_1,\beta_1\}$, respectively. Suppose that $f$ and $\tilde{f}$ are ICWE flows of the network before and after the information expansion. First, we must have $\ell_{\alpha_1}(f)>\ell_{\beta_1}(f)$. Otherwise, $f$ would be an ICWE flow of the network after the information expansion, and IBP would not occur. So $f_{\alpha_1}^{(1,1)}=0$. Because the path selection of type-$(1,0)$ travelers changes before and after the information expansion, we have $r_{(1,0)}=f_{\alpha_1}^{(1,0)}>\tilde{f}_{\alpha_1}^{(1,0)}$. We know that $f_{\alpha_1}=f_{\alpha_1}^{(1,0)}+f_{\alpha_1}^{(1,1)}=r_{(1,0)}+0=r_{(1,0)}$ and $f_{\beta_1}=f_{\beta_1}^{(1,1)}=r_{(1,1)}$.

Next, we will prove that $r_{(1,0)}=f_{\alpha_1}>\tilde{f}_{\alpha_1}$. Assume by contradiction that $\tilde{f}_{\alpha_1}\ge r_{(1,0)}$. We construct a new feasible flow $\bar{f}$ for the network before the information expansion. For $i\ge 2$ and for any $P\in\mathcal P_{(i,1)}$, set $\bar{f}_P^{(i,1)}=\tilde{f}_P^{(i,1)}$. Let $\bar{f}_{\alpha_1}^{(1,0)}=r_{(1,0)}$, $\bar{f}_{\beta_1}^{(1,0)}=0$, $\bar{f}_{\alpha_1}^{(1,1)}=\tilde{f}_{\alpha_1}-r_{(1,0)}$, and $\bar{f}_{\beta_1}^{(1,1)}=\tilde{f}_{\beta_1}$. Then, we have $\bar{f}_{\alpha_1}=\bar{f}_{\alpha_1}^{(1,0)}+\bar{f}_{\alpha_1}^{(1,1)}=\tilde{f}_{\alpha_1}$ and $\bar{f}_{\beta_1}=\bar{f}_{\beta_1}^{(1,0)}+\bar{f}_{\beta_1}^{(1,1)}=\tilde{f}_{\beta_1}$. Thus, for every OD path $P$, $\bar{f}_P=\tilde{f}_P$. So $\bar{f}$ is an ICWE flow in the network after the information expansion. Since $\bar{f}_{\alpha_1}^{(1,0)}=r_{(1,0)}$, $\bar{f}$ is a feasible flow for the network before the information expansion. Thus, $\bar{f}$ is an ICWE flow for the network before the information expansion. By the uniqueness of ICWE, it follows that $\ell_{(1,0)}(f)=\ell_{(1,0)}(\bar{f})$. Since $\bar{f}_{\alpha_1}=\tilde{f}_{\alpha_1}\ge r_{(1,0)}>0$, we derive $\ell_{(1,0)}(f)=\ell_{(1,0)}(\bar{f})=\ell_{\alpha_1}(\bar{f})=\ell_{\alpha_1}(\tilde{f})=\ell_{(1,0)}(\tilde{f})$, which contradicts with the definition of IBP ($\ell_{(1,0)}(f)<\ell_{(1,0)}(\tilde{f})$). Therefore, $f_{\alpha_1}>\tilde{f}_{\alpha_1}$.

Then, we construct a new flow $\hat{f}$ for the network after the information expansion with $\hat{f}_{\alpha_1}^{(1,1)}=0$. For $i\ge2$ and for any $P\in\mathcal P_{(i,1)}$, $\hat{f}_P^{(i,1)}=\tilde{f}_P^{(i,1)}$. $\hat{f}_{\alpha_1}^{(1,0)}=\tilde{f}_{\alpha_1}$, $\hat{f}_{\beta_1}^{(1,0)}=r_{(1,0)}-\tilde{f}_{\alpha_1}$, $\hat{f}_{\alpha_1}^{(1,1)}=0$ and $\hat{f}_{\beta_1}^{(1,1)}=r_{(1,1)}$. Then, we have $\hat{f}_{\alpha_1}=\hat{f}_{\alpha_1}^{(1,0)}+\hat{f}_{\alpha_1}^{(1,1)}=\tilde{f}_{\alpha_1}$ and $\hat{f}_{\beta_1}=\hat{f}_{\beta_1}^{(1,0)}+\hat{f}_{\beta_1}^{(1,1)}=\tilde{f}_{\beta_1}$. Thus, for every OD path $P$, $\hat{f}_P=\tilde{f}_P$. Thus, $\hat{f}$ is an ICWE of the network after the information expansion with $\hat{f}_{\alpha_1}^{(1,1)}=0$. So type-$(1,1)$ travelers choose the same path of the network before and after the information expansion. By the proof of Lemma \ref{lem:type_n+1}, we can derive the instance of IBP that results from excluding the type-$(1,1)$ travelers. 
\end{proof}

By Lemma~\ref{lem:type_n}, We can assume that there are $n$ types of travelers in cycle $C$. Travelers of $n$ types are re-denoted using a simplified notation: type-$i$ for each $i\in [n]$ with $(o_i,d_i)$. Type-$1$ travelers have the expanded information sets $\mathcal P_1=\{\alpha_1\}$ and $\tilde{\mathcal P}_1=\{\alpha_1,\beta_1\}$. Type-$i$ travelers have the complete information sets $\{\alpha_i,\beta_i\}$ for $i\in [n]\setminus\{1\}$. 
In the remainder of this section, we represent an IBP instance using a simplified triple notation $(C,\ell,r)$.
Suppose that $f$ and $\tilde{f}$ are ICWE flows of the network before and after the information expansion. Since $f_{\alpha_i}+f_{\beta_i}=\tilde{f}_{\alpha_i}+\tilde{f}_{\beta_i}$, we assume $\alpha_i$ satisfies $f_{\alpha_i} \ge \tilde{f}_{\alpha_i}$ and $\beta_i$ satisfies $f_{\beta_i} \le \tilde{f}_{\beta_i}$ for each $i\in [n]$. 
The subsequent lemma indicates that it is sufficient to consider the case where $f_{\beta_i}=0$ and $\tilde{f}_{\alpha_i}=0$ for each $i\in [n]$.

\begin{lemma}\label{lem:cir_change_path}
Suppose that $C$ is a cycle with $n$ OD pairs. If there is an IBP instance $(C,\ell,r)$ with $n$ types of travelers in which all travelers  except type-$1$ travelers 
 have the whole cycle $C$ as their information set, then there is an IBP instance $(C, \ell',r')$ with $f_{\beta_i}=0$ and $\tilde{f}_{\alpha_i}=0$ for each $i\in [n]$.
\end{lemma}
\begin{proof}
Let $r'_i=f_{\alpha_i}-\tilde{f}_{\alpha_i}=\tilde{f}_{\beta_i}-f_{\beta_i}$ and $\ell'_e(x)=\ell_e(x+\sum_{e\in\alpha_i}\tilde{f}_{\alpha_i}+\sum_{e\in\beta_i}f_{\beta_i})$ for each $i\in[n]$ and $e\in E(C)$. Let $f'_{\alpha_i}=\tilde{f}'_{\beta_i}=r'_i$ and $f'_{\beta_i}=\tilde{f}'_{\alpha_i}=0$. We have $\ell_e(f_e)=\ell'_e(f'_e)$ and $ \ell_e(\tilde{f}_e)=\ell'_e(\tilde{f}'_e)$ for each $e\in E(C)$. Then we have $f'$  and $\tilde{f}'$ are ICWE flows of $(G,\ell', r')$ before and after the information expansion. Since $(G, \ell, r)$ is an IBP instance, $(G, \ell', r')$ is an IBP instance with $f_{\beta_i}=0$ and $\tilde{f}_{\alpha_i}=0$ for each $i\in [n]$.
\end{proof}

By Lemmas~\ref{lem:type_n} and \ref{lem:cir_change_path}, we have that if there is an IBP instance $(C, \ell, r)$ with $n$ OD pairs, then there must be an IBP instance with $n$ types of travelers and each type of  travelers chooses different paths before and after the information expansion. So we only need to consider this situation: $C$ is a cycle network with $n$ OD pairs and there are $n$ types of travelers. For each $i\in [n]$, type-$i$ travelers' OD pair is $(o_i, d_i)$ and all of them choose $o_i$-$d_i$ path $\alpha_i$ before the information expansion and $\beta_i$ after the information expansion. 

Recall that $f$ and $\tilde{f}$ are ICWE flows before and after the information expansion. Simplified notations: 
\begin{itemize}
\item
For each edge $e$, $\ell_e=\ell_e(f_e)$ and $\tilde{\ell}_e=\ell_e(\tilde{f}_e)$.
\item
For each $\alpha_i$, $\ell_{\alpha_i}=\ell_{\alpha_i}(f)$ and $\tilde{\ell}_{\alpha_i}=\ell_{\alpha_i}(\tilde{f})$. 
\item
For each $\beta_i$, $\ell_{\beta_i}=\ell_{\beta_i}(f)$ and $\tilde{\ell}_{\beta_i}=\ell_{\beta_i}(\tilde{f})$. 
\end{itemize}

Under the circumstances where IBP occurs, we establish the following inequalities, with the explanation for each detailed after the colon.
\begin{enumerate}
\item
$\ell_{\beta_1}<\ell_{\alpha_1}$: Otherwise, after the information expansion, the path choices of type-$1$ travelers will not change, which contradicts with the occurrence of IBP.
\item
$\ell_{\alpha_1}<\tilde{\ell}_{\beta_1}$: From the definition of IBP, the latency of ICWE after the information expansion is strictly greater than the latency of ICWE before the information expansion.
\item
$\ell_{\alpha_i}\le \ell_{\beta_i}$ for each $2\le i\le n$: It can be derived from $f$ being the ICWE before the information expansion and the type-$i$ travelers have complete information for each $2 \leq i \leq n$.
\item
$\tilde{\ell}_{\beta_i}\le \tilde{\ell}_{\alpha_i}$ for each $1\le i\le n$: It follows from $\tilde{f}$ being the ICWE after the information expansion and the type-$i$ travelers have complete information for each $1 \leq i \leq n$.
\end{enumerate}

From the above properties, we can deduce the subsequent key lemma:
\begin{lemma}\label{lem:key}
If IBP occurs, then $\ell_{\alpha_i}<\tilde{\ell}_{\alpha_i}$ for each $i\in [n]$.
\end{lemma}

\begin{proof}
According to the above inequalities (2) and (4), we have $\ell_{\alpha_1}<\tilde{\ell}_{\beta_1}\le\tilde{\ell}_{\alpha_1}$. By the inequalities (1), (2) and (4), we derive
$$\ell_{\alpha_1}+\ell_{\beta_1}<2\ell_{\alpha_1}<2\tilde{\ell}_{\beta_1}\le\tilde{\ell}_{\alpha_1}+\tilde{\ell}_{\beta_1}.$$
Notice that $\ell_{\alpha_i}+\ell_{\beta_i}=\ell_{\alpha_1}+\ell_{\beta_1}$ and $\tilde{\ell}_{\alpha_i}+\tilde{\ell}_{\beta_i}=\tilde{\ell}_{\alpha_1}+\tilde{\ell}_{\beta_1}$ for each $2\le i\le n$. According to the inequalities (3) and (4), for each $2\le i\le n$, we have 
$$\ell_{\alpha_i}\le\frac12(\ell_{\alpha_i}+\ell_{\beta_i})=\frac12(\ell_{\alpha_1}+\ell_{\beta_1})<\frac12(\tilde{\ell}_{\alpha_1}+\tilde{\ell}_{\beta_1})=\frac12(\tilde{\ell}_{\alpha_i}+\tilde{\ell}_{\beta_i})\le\tilde{\ell}_{\alpha_i},$$
which completes the proof.
\end{proof}

For the traffic rate $r_i$ of type-$i$ travelers, the following lemma is established.
\begin{lemma}\label{lem:r_i}
If IBP occurs, then $r_i<\sum_{j\neq i}r_j$ for each $i\in [n]$.
\end{lemma}
\begin{proof}
Suppose there exists $i\in [n]$ such that $r_i\ge\sum_{j\neq i}r_j$. Consider the path $\alpha_i$. For arbitrary edge $e\in\alpha_i$, $f_e\ge r_i$ and $\tilde{f}_e\le\sum_{j\neq  i}r_j$ hold, which means $f_e\ge \tilde{f}_e$ for each $e\in\alpha_i$. Due to the monotonicity of the latency function, $\ell_e\ge \tilde{\ell}_e$ for each $e\in\alpha_i$. Then $\ell_{\alpha_i}=\sum_{e\in\alpha_i}\ell_e\ge\sum_{e\in\alpha_i}\tilde{\ell}_e=\tilde{\ell}_{\alpha_i}$, which contradicts with the result of Lemma \ref{lem:key}.
\end{proof}

\begin{lemma}\label{lem:2od}
A cycle $C$ with two OD pairs is immune to IBP.
\end{lemma}
\begin{proof}
If IBP occurs, by Lemma \ref{lem:r_i}, the inequalities $r_1<r_2$ and $r_2<r_1$ hold simultaneously, which is a contradiction.
\end{proof}

\begin{lemma}\label{lem:union_cycle}
If there is an IBP instance $(C,\ell,r)$ with $n~(\ge 3)$ types of travelers, which satisfies one of the following conditions:
\begin{enumerate}
\item
$\exists~i,j$, $\beta_i\cup\beta_j=C$;
\item
$\exists~i\neq1,j\neq1$, $\alpha_i\cup\alpha_j=C$,
\end{enumerate}
then we can construct a new IBP instance with $n-1$ types of travelers.
\end{lemma}
\begin{proof}
Suppose there exists $i$ and $j$ with $\beta_i\cup\beta_j=C$. Since $\alpha_i\cup\beta_i=C$ and $\alpha_j\cup\beta_j=C$, then $\alpha_i\subset\beta_j$ and $\alpha_j\subset\beta_i$. Notice that $\alpha_i=\beta_j\setminus\beta_i$ and $\alpha_j=\beta_i\setminus\beta_j$. Then we have $\tilde{\ell}_{\beta_i}\le\tilde{\ell}_{\alpha_i}=\tilde{\ell}_{\beta_j\setminus\beta_i}\le\tilde{\ell}_{\beta_j}$ and $\tilde{\ell}_{\beta_j}\le\tilde{\ell}_{\alpha_j}=\tilde{\ell}_{\beta_i\setminus\beta_j}\le\tilde{\ell}_{\beta_i}$. Thus, $\tilde{\ell}_{\alpha_i}=\tilde{\ell}_{\beta_i}=\tilde{\ell}_{\alpha_j}=\tilde{\ell}_{\beta_j}$ and $\tilde{\ell}_{\beta_i\cap\beta_j}=0$. 

Assume, without loss of generality, that $r_i\le r_j$. Construct a new flow $\bar{f}$: For $k\in[n]\setminus\{i,j\}$, $\bar{f}_{\alpha_k}=0$ and $\bar{f}_{\beta_k}=r_k$; $\bar{f}_{\alpha_i}=r_i$, $\bar{f}_{\beta_i}=0$, $\bar{f}_{\alpha_j}=r_i$ and $\bar{f}_{\beta_j}=r_j-r_i$. Thus $\bar{f}_e=\tilde{f}_e$ for any $e\in \alpha_i\cup\alpha_j$ and $\bar{f}_e\le\tilde{f}_e$ for any $e\in\beta_i\cap\beta_j$. Then, we derive that $\ell_e(\bar{f}_e)=\tilde{\ell}_e$ for any $e\in C$. Therefore, $\bar{f}$ is an ICWE flow after the information expansion. Then, we know that type-$i$ travelers choose the same path of the network before and after the information expansion. If $i=1$, then $\bar{f}$ is an ICWE flow in the network before the information expansion. Then we derive that $\ell_{\alpha_1}=\ell_{\alpha_1}(\bar{f})=\tilde{\ell}_{\alpha_1}$, which contradicts with Lemma \ref{lem:key}. Thus, $i\neq1$. By the proof of Lemma \ref{lem:type_n+1}, we can derive the instance of IBP with $n-1$ types travelers. 

{The proof of the situation $\exists~i\neq1,j\neq1$, $\alpha_i\cup\alpha_j=C$ is similar.}

Suppose there exists $i\neq1$ and $j\neq1$ with $\alpha_i\cup\alpha_j=C$. Since $\alpha_i\cup\beta_i=C$ and $\alpha_j\cup\beta_j=C$, then $\beta_i\subset\alpha_j$ and $\beta_j\subset\alpha_i$. Notice that $\beta_i=\alpha_j\setminus\alpha_i$ and $\beta_j=\alpha_i\setminus\alpha_j$. Then we have $\ell_{\alpha_i}\le\ell_{\beta_i}=\ell_{\alpha_j\setminus\alpha_i}\le\ell_{\alpha_j}$ and $\ell_{\alpha_j}\le\ell_{\beta_j}=\ell_{\alpha_i\setminus\alpha_j}\le\ell_{\alpha_i}$. Thus, $\ell_{\beta_i}=\ell_{\alpha_i}=\ell_{\beta_j}=\ell_{\alpha_j}$ and $\ell_{\alpha_i\cap\alpha_j}=0$. 

Assume, without loss of generality, that $r_i\le r_j$. Construct a new flow $\hat{f}$: For $k\in[n]\setminus\{i,j\}$, $\hat{f}_{\alpha_k}=r_k$ and $\hat{f}_{\beta_k}=0$; $\hat{f}_{\alpha_i}=0$, $\hat{f}_{\beta_i}=r_i$, $\hat{f}_{\alpha_j}=r_j-r_i$ and $\hat{f}_{\beta_j}=r_i$. Thus $\hat{f}_e=f_e$ for any $e\in \beta_i\cup\beta_j$ and $\hat{f}_e\le f_e$ for any $e\in\alpha_i\cap\alpha_j$. Then, we derive that $\ell_e(\hat{f}_e)=\ell_e$ for any $e\in C$. Therefore, $\hat{f}$ is an ICWE flow before the information expansion. Then, we know that type-$i$ travelers choose the same path of the network before and after the information expansion. By the proof of Lemma \ref{lem:type_n+1}, we can derive the instance of IBP with $n-1$ types travelers. 
\end{proof}

\begin{lemma}\label{lem:3od}
A cycle $C$ with three OD pairs is immune to IBP.
\end{lemma}

\begin{proof}
Assume that IBP occurs by contradiction. We will discuss the following scenarios.
\begin{enumerate}
\item
There exist $i$ and $j$ such that $\alpha_i\subset \alpha_j$.

For any $e\in\alpha_i$, according to Lemma \ref{lem:r_i}, we derive that $f_e\ge r_i+r_j>r_k\ge\tilde{f}_e$ where $k=\{1,2,3\}\setminus\{i,j\}$. Then, it follows that $\ell_{\alpha_i}\ge\tilde{\ell}_{\alpha_i}$, which contradicts with Lemma \ref{lem:key}.
\item
There exist $i$ and $j$ such that $\alpha_i\cap\alpha_j=\emptyset$.

Then we have $\beta_i\cup\beta_j=C$. By Lemma \ref{lem:union_cycle}, we can construct a new IBP instance with $2$ types of travelers, which contradicts with Lemma \ref{lem:2od}.
\item
There exist $i$ and $j$ such that $\alpha_i\cup\alpha_j=C$.

If $\{i,j\}=\{2,3\}$, according to Lemma \ref{lem:union_cycle}, we can construct a new IBP instance with $2$ types of travelers, which contradicts with Lemma \ref{lem:2od}.

Otherwise, suppose that $j>i=1$ and let $k=\{2,3\}\setminus\{j\}$. For any $e\in\alpha_k$, according to Lemma \ref{lem:r_i}, we deduce that $f_e\ge r_k+\min\{r_1,r_j\}>\max\{r_1,r_j\}\ge \tilde{f}_e$. Then, we show that $\ell_{\alpha_k}\ge\tilde{\ell}_{\alpha_k}$, which contradicts with Lemma \ref{lem:key}.
\item
For any $i$ and $j$ such that $\alpha_i\setminus\alpha_j\neq\emptyset$, $\alpha_i\cap\alpha_j\neq\emptyset$ and $\alpha_i\cup\alpha_j\ne C$.

\begin{figure}[h]
\begin{center}
\includegraphics[scale=0.36]{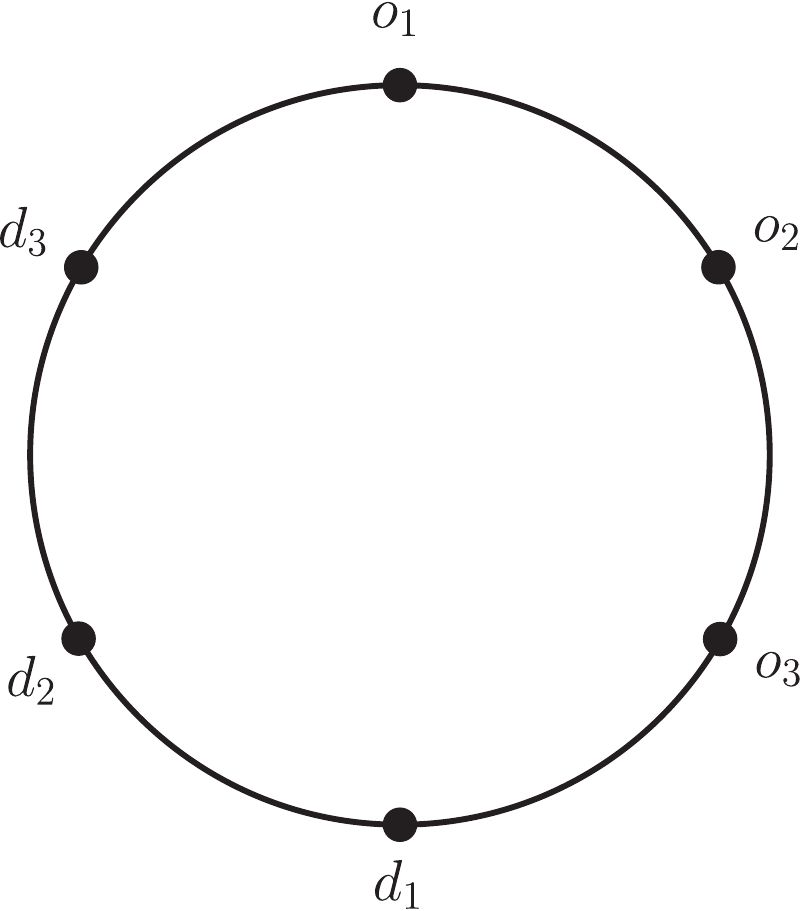}
\caption{The symmetric structured network with three OD pairs}\label{fig:3od_0}
\end{center}
\end{figure}

In this scenario, the network can only consist of the following structure with three OD pairs as shown in Figure~\ref{fig:3od_0}. In such a network structure, due to symmetry, we can assume that $\alpha_1$ represents the clockwise direction from $o_1$ to $d_1$. There are four possible classifications for the paths of $\alpha_2$ and $\alpha_3$. In the following three cases as shown in Figure~\ref{fig:3od_3cases}, for any $e\in\alpha_i$, according to Lemma \ref{lem:r_i}, we derive that $f_e\ge r_i+\min\{r_j,r_k\}>\max\{r_j,r_k\}\ge\tilde{f}_e$ where $\{i,j,k\}=\{1,2,3\}$. Then, it follows that $\ell_{\alpha_i}\ge\tilde{\ell}_{\alpha_i}$, which contradicts with Lemma \ref{lem:key}. The last case is shown in Figure~\ref{fig:3od_4}. Since $\ell_{\alpha_2}\le\ell_{\beta_2}$ and $\ell_{\alpha_3}\le\ell_{\beta_3}$, $\ell_{e_1}+\ell_{e_5}+\ell_{e_6}\le\ell_{e_2}+\ell_{e_3}+\ell_{e_4}$ and $\ell_{e_3}+\ell_{e_4}+\ell_{e_5}\le\ell_{e_1}+\ell_{e_2}+\ell_{e_6}$. So $\ell_{e_5}\le\ell_{e_2}$. Similarly, since $\tilde{\ell}_{\beta_2}\le\tilde{\ell}_{\alpha_2}$ and $\tilde{\ell}_{\beta_3}\le\tilde{\ell}_{\alpha_3}$, we show that $\tilde{\ell}_{e_2}\le\tilde{\ell}_{e_5}$. According to Lemma \ref{lem:r_i}, $f_{e_5}=r_2+r_3>r_1=\tilde{f}_{e_5}$, which means that $\ell_{e_5}\ge\tilde{\ell}_{e_5}$. Similarly, we reveal that $\ell_{e_1}\ge\tilde{\ell}_{e_1}$ and $\ell_{e_3}\ge\tilde{\ell}_{e_3}$. By Lemma \ref{lem:key}, $\ell_{e_1}+\ell_{e_2}+\ell_{e_3}=\ell_{\alpha_1}<\tilde{\ell}_{\alpha_1}=\tilde{\ell}_{e_1}+\tilde{\ell}_{e_2}+\tilde{\ell}_{e_3}$. Then we deduce that $\tilde{\ell}_{e_2}>\ell_{e_2}$. Combining the above inequalities, $\tilde{\ell}_{e_5}\ge\tilde{\ell}_{e_2}>\ell_{e_2}\ge\ell_{e_5}\ge\tilde{\ell}_{e_5}$, which is a contradiction.
\end{enumerate}
\begin{figure}[h]
\centering
\begin{subfigure}[b]{0.3\textwidth}
\centering
\includegraphics[width=\textwidth, trim = 0cm 0cm 0cm 0cm, clip]{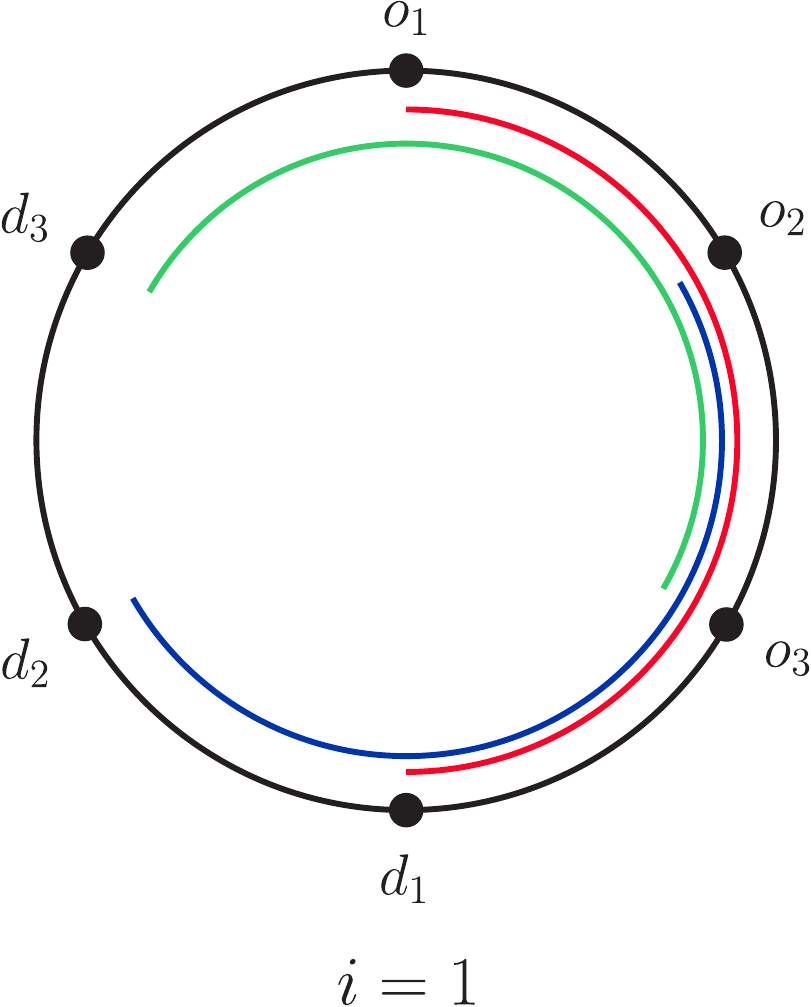}
\end{subfigure}
\hspace{2mm}
\begin{subfigure}[b]{0.3\textwidth}
\centering
\includegraphics[width=\textwidth, trim = 0cm 0cm 0cm 0cm, clip]{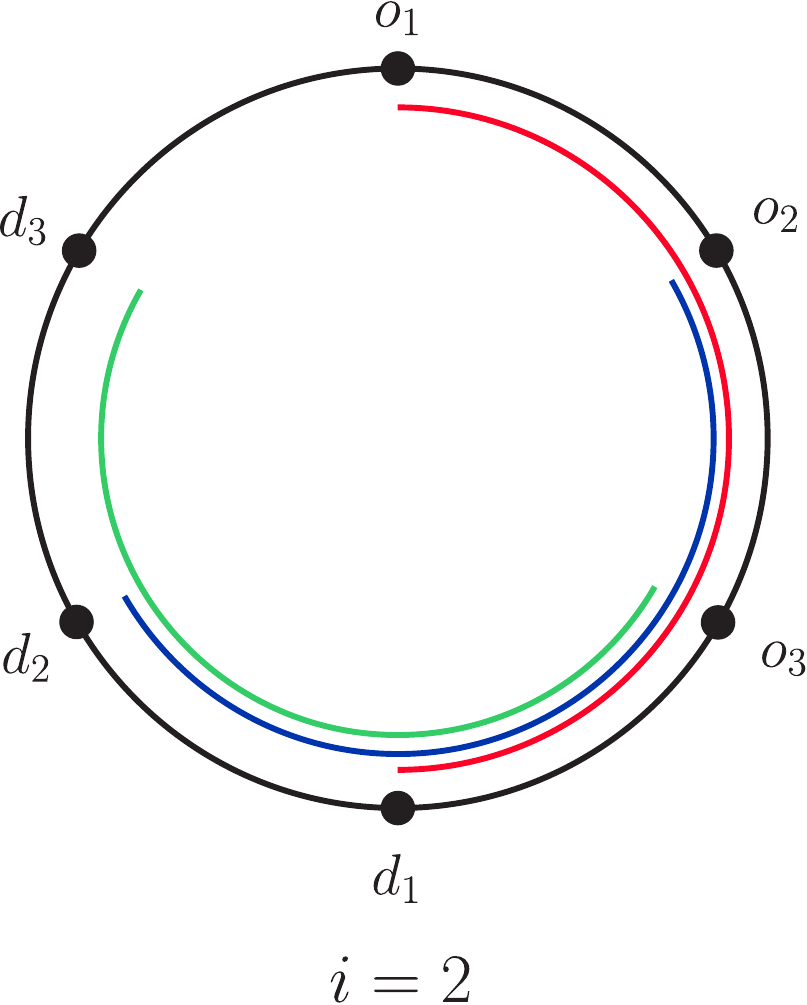}
\end{subfigure}
\hspace{2mm}
\begin{subfigure}[b]{0.3\textwidth}
\centering
\includegraphics[width=\textwidth, trim = 0cm 0cm 0cm 0cm, clip]{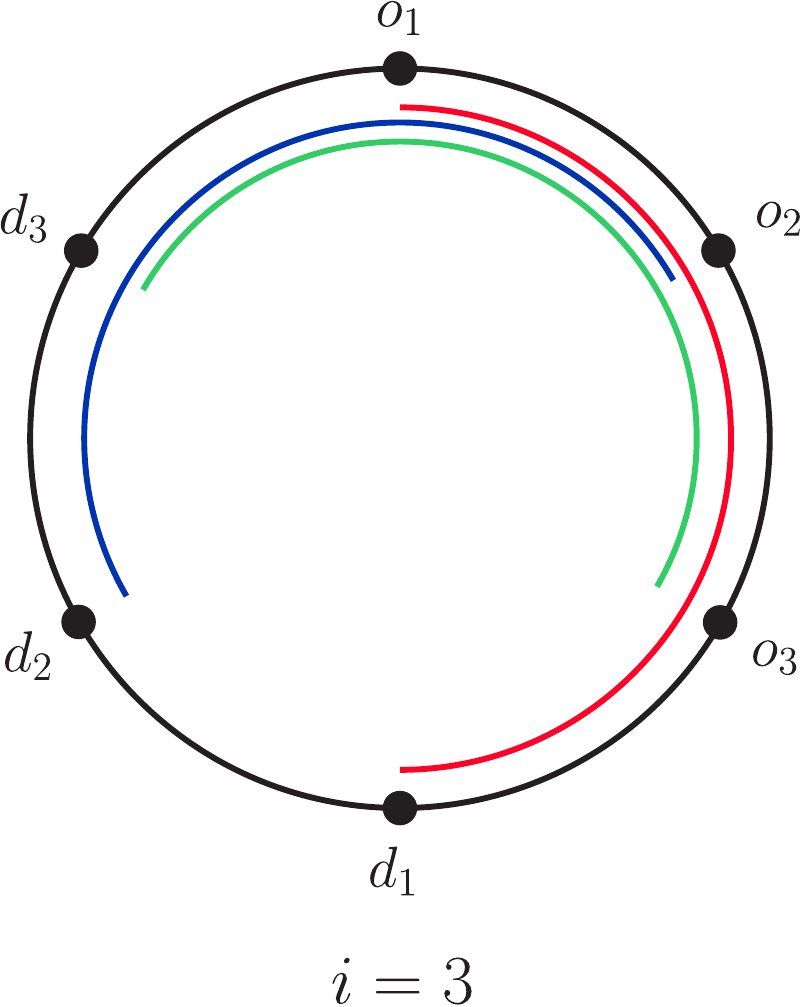}
\end{subfigure}
\caption{Three cases of the symmetric structured network with three OD pairs}
\label{fig:3od_3cases}
\end{figure}
\begin{figure}[h]
\begin{center}
\includegraphics[scale=0.36]{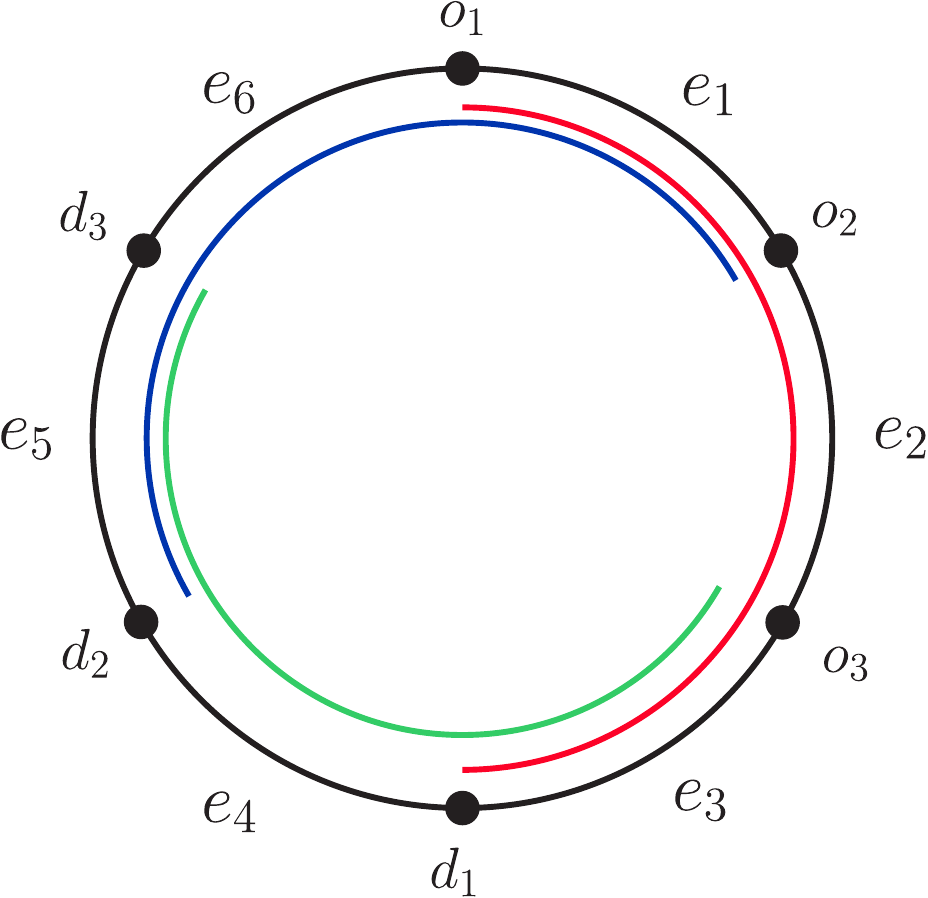}
\caption{The fourth case of the symmetric structured network with three OD pairs}\label{fig:3od_4}
\end{center}
\end{figure}
Therefore, a cycle with three OD pairs is IBP-free.
\end{proof}

Similar to the fourth case in the proof of Lemma \ref{lem:3od}, we consider a structure on a cycle with $n$ OD pairs: a completely symmetric distribution of terminal vertices, as shown in Figure~\ref{fig:sym-n-od}, which satisfies the following conditions:
For any $i$ and $j$ such that $\alpha_i\setminus\alpha_j\neq\emptyset$, $\alpha_i\cap\alpha_j\neq\emptyset$ and $\alpha_i\cup\alpha_j\ne C$. 

\begin{figure}[h]
\begin{center}
\includegraphics[scale=0.36]{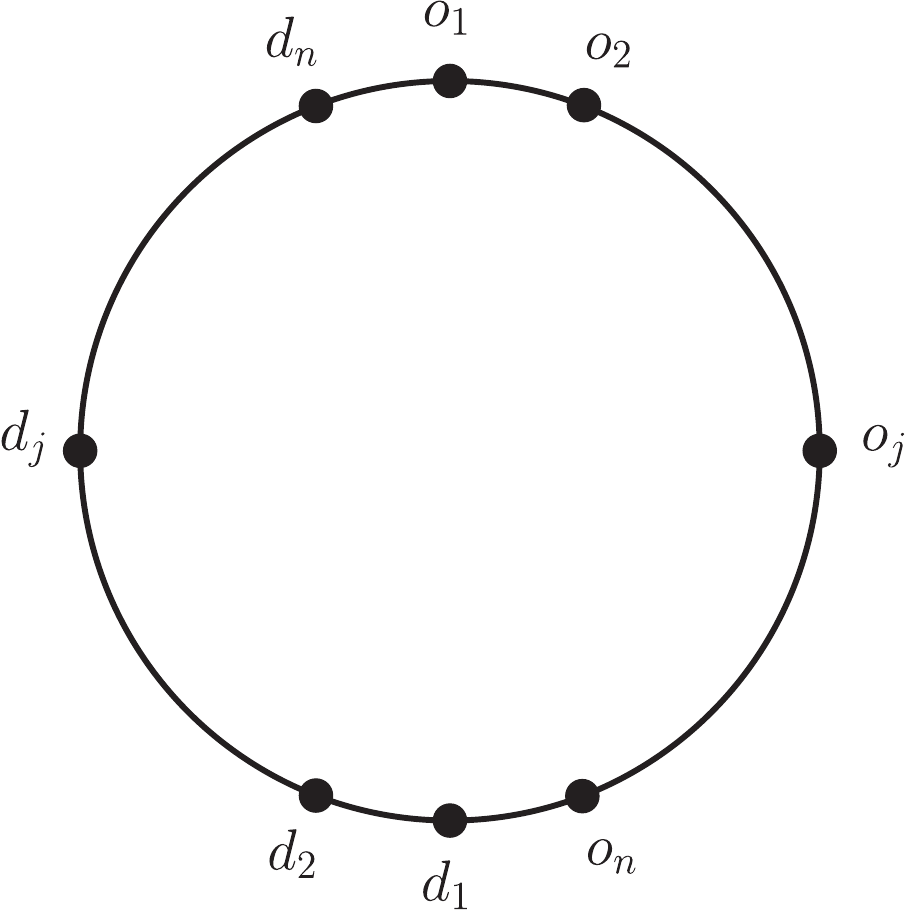}
\caption{The symmetric structured network with $n$ OD pairs}\label{fig:sym-n-od}
\end{center}
\end{figure}

We can represent the paths of type-$i$ travelers by defining an index as follows: For any $i\in[n]$, define 
\[
\sigma(i)=\begin{cases}
1 & \text{$\alpha_i$ is the clockwise path from $o_i$ to $d_i$.}\\
-1 & \text{$\alpha_i$ is the counterclockwise path from $o_i$ to $d_i$.}\\
\end{cases}
\]

Without loss of generality, we assume that $\sigma(1)=1$.

\begin{lemma}\label{lem:4equal}
Consider $n\ge3$ and any $i\neq1,j\neq1$. Let $e_1=o_io_j$, $e_2=d_id_j$ and $r=\sum_{i=1}^nr_i$.
\begin{enumerate}
\item
If $\sigma(i)=1$, $\sigma(j)=-1$ and  $f_{e_1}\ge\frac{r}{2}\ge f_{e_2}$, then 
$\ell_{e_1}=\ell_{e_2}=\tilde{\ell}_{e_1}=\tilde{\ell}_{e_2}$.
\item
If $\sigma(i)=-1$, $\sigma(j)=1$ and $f_{e_1}\le\frac{r}{2}\le f_{e_2}$, then 
$\ell_{e_1}=\ell_{e_2}=\tilde{\ell}_{e_1}=\tilde{\ell}_{e_2}$.
\end{enumerate}
\end{lemma}

\begin{proof}
1. If $\sigma(i)=1$ and $\sigma(j)=-1$, then $\alpha_i\cap\alpha_j=\{e_1\}$ and $C\setminus(\alpha_i\cup\alpha_j)=\{e_2\}$. We deduce that $\ell_{\alpha_i}+\ell_{\alpha_j}=\sum_{e\in C}\ell_e+\ell_{e_1}-\ell_{e_2}$ and $\ell_{\beta_i}+\ell_{\beta_j}=\sum_{e\in C}\ell_e-\ell_{e_1}+\ell_{e_2}$. Since $\ell_{\alpha_i}\le\ell_{\beta_i}$ and $\ell_{\alpha_j}\le\ell_{\beta_j}$, $\sum_{e\in C}\ell_e+\ell_{e_1}-\ell_{e_2}\le\sum_{e\in C}\ell_e-\ell_{e_1}+\ell_{e_2}$. Thus, we have $\ell_{e_1}\le\ell_{e_2}$. Similarly, we derive $\tilde{\ell}_{e_1}\ge\tilde{\ell}_{e_2}$. As $f_{e_1}\ge\frac{r}{2}\ge\tilde{f}_{e_1}$, it follows $\ell_{e_1}\ge\tilde{\ell}_{e_1}$. As $f_{e_2}\le\frac{r}{2}\le\tilde{f}_{e_2}$, it follows $\ell_{e_2}\le\tilde{\ell}_{e_2}$. Combining the above inequalities, we obtain that $\ell_{e_1}=\ell_{e_2}=\tilde{\ell}_{e_1}=\tilde{\ell}_{e_2}$.

2. If $\sigma(i)=-1$ and $\sigma(j)=1$, then $\alpha_i\cap\alpha_j=\{e_2\}$ and $C\setminus(\alpha_i\cup\alpha_j)=\{e_1\}$. We deduce that $\ell_{\alpha_i}+\ell_{\alpha_j}=\sum_{e\in C}\ell_e+\ell_{e_2}-\ell_{e_1}$ and $\ell_{\beta_i}+\ell_{\beta_j}=\sum_{e\in C}\ell_e-\ell_{e_2}+\ell_{e_1}$. Since $\ell_{\alpha_i}\le\ell_{\beta_i}$ and $\ell_{\alpha_j}\le\ell_{\beta_j}$, $\sum_{e\in C}\ell_e+\ell_{e_2}-\ell_{e_1}\le\sum_{e\in C}\ell_e-\ell_{e_2}+\ell_{e_1}$. Thus, we have $\ell_{e_2}\le\ell_{e_1}$. Similarly, we derive $\tilde{\ell}_{e_2}\ge\tilde{\ell}_{e_1}$. As $f_{e_1}\le\frac{r}{2}\le\tilde{f}_{e_1}$, it follows $\ell_{e_1}\le\tilde{\ell}_{e_1}$. As $f_{e_2}\ge\frac{r}{2}\ge\tilde{f}_{e_2}$, it follows $\ell_{e_2}\ge\tilde{\ell}_{e_2}$. Combining the above inequalities, we obtain that $\ell_{e_1}=\ell_{e_2}=\tilde{\ell}_{e_1}=\tilde{\ell}_{e_2}$.
\end{proof}

\begin{lemma}\label{lem:4equal_2}
Consider $n\ge3$ and any $i\neq1,j\neq1$. Let $e_1=o_io_j$, $e_2=d_id_1$, $e_3=d_1d_j$ and $r=\sum_{i=1}^nr_i$.
\begin{enumerate}
\item
If $\sigma(i)=1$, $\sigma(j)=-1$,   $f_{e_1}\ge\frac{r}{2}$ and $f_{e_2},~f_{e_3}\le\frac{r}{2}$, then 
$\ell_{e_1}=\tilde{\ell}_{e_1}=\ell_{e_2}+\ell_{e_3}=\tilde{\ell}_{e_2}+\tilde{\ell}_{e_3}$.
\item
If $\sigma(i)=-1$, $\sigma(j)=1$, $f_{e_1}\le\frac{r}{2}$ and $f_{e_2},~f_{e_3}\ge\frac{r}{2}$, then 
$\ell_{e_1}=\tilde{\ell}_{e_1}=\ell_{e_2}+\ell_{e_3}=\tilde{\ell}_{e_2}+\tilde{\ell}_{e_3}$.
\end{enumerate}
\end{lemma}

\begin{proof}
1. If $\sigma(i)=1$ and $\sigma(j)=-1$, then $\alpha_i\cap\alpha_j=\{e_1\}$ and $C\setminus(\alpha_i\cup\alpha_j)=\{e_2,e_3\}$. We deduce that $\ell_{\alpha_i}+\ell_{\alpha_j}=\sum_{e\in C}\ell_e+\ell_{e_1}-\ell_{e_2}-\ell_{e_3}$ and $\ell_{\beta_i}+\ell_{\beta_j}=\sum_{e\in C}\ell_e-\ell_{e_1}+\ell_{e_2}+\ell_{e_3}$. Since $\ell_{\alpha_i}\le\ell_{\beta_i}$ and $\ell_{\alpha_j}\le\ell_{\beta_j}$, $\sum_{e\in C}\ell_e+\ell_{e_1}-\ell_{e_2}-\ell_{e_3}\le\sum_{e\in C}\ell_e-\ell_{e_1}+\ell_{e_2}+\ell_{e_3}$. Thus, we have $\ell_{e_1}\le\ell_{e_2}+\ell_{e_3}$. Similarly, we derive $\tilde{\ell}_{e_1}\ge\tilde{\ell}_{e_2}+\tilde{\ell}_{e_3}$. As $f_{e_1}\ge\frac{r}{2}\ge\tilde{f}_{e_1}$, it follows $\ell_{e_1}\ge\tilde{\ell}_{e_1}$. As $f_{e_2}\le\frac{r}{2}\le\tilde{f}_{e_2}$ and $f_{e_3}\le\frac{r}{2}\le\tilde{f}_{e_3}$, it follows $\ell_{e_2}+\ell_{e_3}\le\tilde{\ell}_{e_2}+\tilde{\ell}_{e_3}$. Combining the above inequalities, we obtain that $\ell_{e_1}=\tilde{\ell}_{e_1}=\ell_{e_2}+\ell_{e_3}=\tilde{\ell}_{e_2}+\tilde{\ell}_{e_3}$.

2. If $\sigma(i)=-1$ and $\sigma(j)=1$, then $\alpha_i\cap\alpha_j=\{e_2,e_3\}$ and $C\setminus(\alpha_i\cup\alpha_j)=\{e_1\}$. We deduce that $\ell_{\alpha_i}+\ell_{\alpha_j}=\sum_{e\in C}\ell_e+\ell_{e_2}+\ell_{e_3}-\ell_{e_1}$ and $\ell_{\beta_i}+\ell_{\beta_j}=\sum_{e\in C}\ell_e-\ell_{e_2}-\ell_{e_3}+\ell_{e_1}$. Since $\ell_{\alpha_i}\le\ell_{\beta_i}$ and $\ell_{\alpha_j}\le\ell_{\beta_j}$, $\sum_{e\in C}\ell_e+\ell_{e_2}+\ell_{e_3}-\ell_{e_1}\le\sum_{e\in C}\ell_e-\ell_{e_2}-\ell_{e_3}+\ell_{e_1}$. Thus, we have $\ell_{e_2}+\ell_{e_3}\le\ell_{e_1}$. Similarly, we derive $\tilde{\ell}_{e_2}+\tilde{\ell}_{e_3}\ge\tilde{\ell}_{e_1}$. As $f_{e_1}\le\frac{r}{2}\le\tilde{f}_{e_1}$, it follows $\ell_{e_1}\le\tilde{\ell}_{e_1}$. As $f_{e_2}\ge\frac{r}{2}\ge\tilde{f}_{e_2}$ and $f_{e_3}\ge\frac{r}{2}\ge\tilde{f}_{e_3}$, it follows $\ell_{e_2}+\ell_{e_3}\ge\tilde{\ell}_{e_2}+\tilde{\ell}_{e_3}$. Combining the above inequalities, we obtain that $\ell_{e_1}=\tilde{\ell}_{e_1}=\ell_{e_2}+\ell_{e_3}=\tilde{\ell}_{e_2}+\tilde{\ell}_{e_3}$.
\end{proof}

Next, we introduce three edge contraction lemmas.

\begin{lemma}\label{lem:shrink_1}
Given an IBP instance $(C,\ell,r)$ with $n~(\ge 3)$ types of travelers, if there exist two edges $e_1=o_io_j$ and $e_2=d_id_j$ with $i\neq1$ and $j\neq1$ satisfying the following conditions:
\begin{enumerate}
\item
$|\alpha_k\cap\{e_1,e_2\}|=1$ for any $k\in [n]$;
\item
$\ell_{e_1}=\ell_{e_2}=\tilde{\ell}_{e_1}=\tilde{\ell}_{e_2}$,
\end{enumerate}
then we can construct a new IBP instance with $n-1$ types of travelers.
\end{lemma}

\begin{proof}
Suppose that $f$ and $\tilde{f}$ are ICWE flows of the network before and after the information expansion. By shrinking $e_1$ and $e_2$ on $C$, we obtain a resulting cycle $C'$. The new labels for the contracted vertices are denoted as $o_{i'}$ and $d_{i'}$ after shrinking $e_1$ and $e_2$, respectively. Let $\alpha_k'$ be the restricted path of $\alpha_k$ on $C'$ and $\beta_k'$ be the restricted path of $\beta_k$ on $C'$ for $k\in[n]$. We have that $\ell_{\alpha_k'}=\ell_{\alpha_k}-\ell_{e_1}$, $\ell_{\beta_k'}=\ell_{\beta_k}-\ell_{e_1}$, $\tilde{\ell}_{\alpha_k'}=\tilde{\ell}_{\alpha_k}-\ell_{e_1}$ and $\tilde{\ell}_{\beta_k'}=\tilde{\ell}_{\beta_k}-\ell_{e_1}$ for $k\in[n]$.
\begin{itemize}
\item
For type-$1$, $\tilde{\ell}_{\beta_1'}\le\tilde{\ell}_{\alpha_1'}$ and $\ell_{\alpha_1'}<\tilde{\ell}_{\beta_1'}$.
\item
For type-$k$ where $k\in[n]\setminus\{1,i,j\}$, $\tilde{\ell}_{\beta_k'}\le\tilde{\ell}_{\alpha_k'}$ and $\ell_{\alpha_k'}\le\ell_{\beta_k'}$.
\item
The type-$i$ and type-$j$ travelers are removed, and type-$i'$ travelers with $r_{i'}=r_i+r_j$ are added. The type-$i'$ travelers choose $\alpha_i'$ with $r_i$ and $\alpha_j'$ with $r_j$ before the information expansion, while they choose $\beta_i'$ with $r_i$ and $\beta_j'$ with $r_j$ after the information expansion. We have that $\tilde{\ell}_{\beta_k'}\le\tilde{\ell}_{\alpha_k'}$ and $\ell_{\alpha_k'}\le\ell_{\beta_k'}$ for $k\in\{i,j\}$.
\end{itemize}
We can conclude that $f$ and $\tilde{f}$ restricted on $C'$ are ICWE flows of $C'$ before and after the information expansion. Thus, we derive a new IBP instance with $n-1$ types of travelers.
\end{proof}

\begin{lemma}\label{lem:shrink_2}
Given an IBP instance $(C,\ell,r)$ with $n~(\ge 3)$ types of travelers, if there exist two edges $e_1=o_io_j$ and $e_2=d_id_j$ with $i\neq1$ and $j\neq1$ satisfying the following conditions:
\begin{enumerate}
\item
$|\alpha_k\cap\{e_1,e_2\}|=1$ for any $k\in [n]\setminus\{1\}$;
\item
$\ell_{e_1}=\ell_{e_2}=\tilde{\ell}_{e_1}=\tilde{\ell}_{e_2}$;
\item
$\beta_1\subset\beta_t\cup\{e_1,e_2\}$ for some $t\neq 1$ and $|\alpha_1\cap\{e_1,e_2\}|\ge1$,
\end{enumerate}
then we can construct a new IBP instance with $n-1$ types of travelers.
\end{lemma}

\begin{proof}
Suppose that $f$ and $\tilde{f}$ are ICWE flows of the network before and after the information expansion. By shrinking $e_1$ and $e_2$ on $C$, we obtain a resulting cycle $C'$. The new labels for the contracted vertices are denoted as $o_{i'}$ and $d_{i'}$ after shrinking $e_1$ and $e_2$, respectively. Let $\alpha_k'$ be the restricted path of $\alpha_k$ on $C'$ and $\beta_k'$ be the restricted path of $\beta_k$ on $C'$ for $k\in[n]$. We have that $\ell_{\alpha_k'}=\ell_{\alpha_k}-\ell_{e_1}$, $\ell_{\beta_k'}=\ell_{\beta_k}-\ell_{e_1}$, $\tilde{\ell}_{\alpha_k'}=\tilde{\ell}_{\alpha_k}-\ell_{e_1}$ and $\tilde{\ell}_{\beta_k'}=\tilde{\ell}_{\beta_k}-\ell_{e_1}$ for $k\in[n]\setminus\{1\}$. In addition, $\ell_{\alpha_1'}\le\ell_{\alpha_1}-\ell_{e_1}$ and $\tilde{\ell}_{\beta_1'}\ge\tilde{\ell}_{\beta_1}-\ell_{e_1}$.
\begin{itemize}
\item
For type-$1$, it follows that $\ell_{\alpha_1'}<\tilde{\ell}_{\beta_1'}$. Since $\beta_1\subset\beta_t\cup\{e_1,e_2\}$, $\beta_1'\subset\beta_t'$ and $\alpha_t'\subset\alpha_1'$. Thus, $\tilde{\ell}_{\beta_1'}\le\tilde{\ell}_{\beta_t'}\le\tilde{\ell}_{\alpha_t'}\le\tilde{\ell}_{\alpha_1'}$.
\item
For type-$k$ where $k\in[n]\setminus\{1,i,j\}$, $\tilde{\ell}_{\beta_k'}\le\tilde{\ell}_{\alpha_k'}$ and $\ell_{\alpha_k'}\le\ell_{\beta_k'}$.
\item
The type-$i$ and type-$j$ travelers are removed, and type-$i'$ travelers with $r_{i'}=r_i+r_j$ are added. The type-$i'$ travelers choose $\alpha_i'$ with $r_i$ and $\alpha_j'$ with $r_j$ before the information expansion, while they choose $\beta_i'$ with $r_i$ and $\beta_j'$ with $r_j$ after the information expansion. We have that $\tilde{\ell}_{\beta_k'}\le\tilde{\ell}_{\alpha_k'}$ and $\ell_{\alpha_k'}\le\ell_{\beta_k'}$ for $k\in\{i,j\}$.
\end{itemize}
We can conclude that $f$ and $\tilde{f}$ restricted on $C'$ are ICWE flows of $C'$ before and after the information expansion. Thus, we derive a new IBP instance with $n-1$ types of travelers.
\end{proof}

\begin{lemma}\label{lem:shrink_3}
Given an IBP instance $(C,\ell,r)$ with $n~(\ge 3)$ types of travelers, if there exist three edges $e_1=o_io_j$, $e_2=d_id_1$ and $e_3=d_1d_j$ with $i\neq1$ and $j\neq1$ satisfying the following conditions:
\begin{enumerate}
\item
For any $k\in [n]\setminus\{1\}$,it follows that $e_1\in\alpha_k$, $e_2,e_3\in\beta_k$ or that $e_2,e_3\in\alpha_k$, $e_1\in\beta_k$;
\item
$\ell_{e_1}=\tilde{\ell}_{e_1}=\ell_{e_2}+\ell_{e_3}=\tilde{\ell}_{e_2}+\tilde{\ell}_{e_3}$;
\item
$\beta_1\subset\beta_t\cup\{e_2,e_3\}$ for some $t\neq 1$ and $e_1\in\alpha_1$,
\end{enumerate}
then we can construct a new IBP instance with $n-1$ types of travelers.
\end{lemma}

\begin{proof}
Suppose that $f$ and $\tilde{f}$ are ICWE flows of the network before and after the information expansion. By shrinking $e_1$ and $e_2$ on $C$, we obtain a resulting cycle $C'$. The new labels for the contracted vertices are denoted as $o_{i'}$ and $d_{i'}$ after shrinking $e_1$ and $e_2$, respectively. We change the latency on edge $e_3$ to $0$ and relabel $d_{j+1}$ as $d_1$. Let $\alpha_k'$ be the restricted path of $\alpha_k$ on $C'$ and $\beta_k'$ be the restricted path of $\beta_k$ on $C'$ for $k\in[n]$. We have that $\ell_{\alpha_k'}=\ell_{\alpha_k}-\ell_{e_1}$, $\ell_{\beta_k'}=\ell_{\beta_k}-\ell_{e_1}$, $\tilde{\ell}_{\alpha_k'}=\tilde{\ell}_{\alpha_k}-\ell_{e_1}$ and $\tilde{\ell}_{\beta_k'}=\tilde{\ell}_{\beta_k}-\ell_{e_1}$ for $k\in[n]\setminus\{1\}$. In addition, $\ell_{\alpha_1'}\le\ell_{\alpha_1}-\ell_{e_1}$ and $\tilde{\ell}_{\beta_1'}\ge\tilde{\ell}_{\beta_1}-\ell_{e_1}$.
\begin{itemize}
\item
For type-$1$, it follows that $\ell_{\alpha_1'}<\tilde{\ell}_{\beta_1'}$. Since $\beta_1\subset\beta_t\cup\{e_1,e_2\}$, $\beta_1'\subset\beta_t'$ and $\alpha_t'\subset\alpha_1'$. Thus, $\tilde{\ell}_{\beta_1'}\le\tilde{\ell}_{\beta_t'}\le\tilde{\ell}_{\alpha_t'}\le\tilde{\ell}_{\alpha_1'}$.
\item
For type-$k$ where $k\in[n]\setminus\{1,i,j\}$, $\tilde{\ell}_{\beta_k'}\le\tilde{\ell}_{\alpha_k'}$ and $\ell_{\alpha_k'}\le\ell_{\beta_k'}$.
\item
The type-$i$ and type-$j$ travelers are removed, and type-$i'$ travelers with $r_{i'}=r_i+r_j$ are added. The type-$i'$ travelers choose $\alpha_i'$ with $r_i$ and $\alpha_j'$ with $r_j$ before the information expansion, while they choose $\beta_i'$ with $r_i$ and $\beta_j'$ with $r_j$ after the information expansion. We have that $\tilde{\ell}_{\beta_k'}\le\tilde{\ell}_{\alpha_k'}$ and $\ell_{\alpha_k'}\le\ell_{\beta_k'}$ for $k\in\{i,j\}$.
\end{itemize}
We can conclude that $f$ and $\tilde{f}$ restricted on $C'$ are ICWE flows of $C'$ before and after the information expansion. Thus, we derive a new IBP instance with $n-1$ types of travelers.
\end{proof}

\begin{lemma}\label{lem:ge_le}
For any $i\in\{1,\ldots,n\}$, if $\sigma(i)=1$, then $f_{o_{i-1}o_i}<f_{o_io_{i+1}}$ where $o_0=d_n$ and $o_{n+1}=d_1$. Otherwise, $f_{o_{i-1}o_i}>f_{o_io_{i+1}}$.
\end{lemma}
\begin{proof}
Notice that $f^j_{o_{i-1}o_i}=f^j_{o_io_{i+1}}$ for any $j\neq i$. If $\sigma(i)=1$, then $f^i_{o_{i-1}o_i}<f^i_{o_io_{i+1}}$. Thus, $f_{o_{i-1}o_i}=\sum_{j=1}^nf^j_{o_{i-1}o_i}<\sum_{j=1}^nf^j_{o_io_{i+1}}=f_{o_io_{i+1}}$. If $\sigma(i)=-1$, then $f^i_{o_{i-1}o_i}>f^i_{o_io_{i+1}}$. Thus, $f_{o_{i-1}o_i}=\sum_{j=1}^nf^j_{o_{i-1}o_i}>\sum_{j=1}^nf^j_{o_io_{i+1}}=f_{o_io_{i+1}}$.
\end{proof}

\begin{lemma}\label{lem:symmetric}
A cycle $C$ shown in Figure~\ref{fig:sym-n-od} with $n~(\ge3)$ OD pairs is immune to IBP.
\end{lemma}

\begin{proof}
(by induction on $n$) When $C$ contains three OD pairs, the conclusion holds by Lemma \ref{lem:3od}. Assuming that the conclusion holds for cases containing fewer than $n$ OD pairs, consider an IBP instance where $C$ contains $n$ OD pairs with a completely symmetric distribution of terminal vertices. Suppose that $f$ and $\tilde{f}$ are ICWE flows of the network before and after the information expansion. Let $r=\sum_{i=1}^nr_i$. Recall that $\sigma(1)=1$. Let us discuss the values of $\sigma(2),\ldots,\sigma(n)$ below:

\noindent{\bf Case 1}. $\sigma(2)=\cdots=\sigma(n)=1$.\\
We derive that $f_{o_nd_1}=\sum_{i=1}^nr_i=r$.
According to Lemma \ref{lem:r_i},
$f_{o_{n-1}o_n}=\sum_{i=1}^{n-1}r_i>\frac{r}{2}$ and $f_{d_1d_2}=\sum_{i=2}^nr_i>\frac{r}{2}$. By Lemma \ref{lem:ge_le}, $f_{o_{i-1}o_i}<f_{o_io_{i+1}}$ and $f_{d_{i-1}d_i}>f_{d_id_{i+1}}$ for any $i\in [n]$ when $o_{0}=d_n$, $o_{n+1}=d_1$, $d_{0}=o_n$ and $d_{n+1}=o_1$.
Let $j=\min\{i\ge1~|~f_{o_io_{i+1}}\ge\frac{r}{2}\}$. Then $1\le j\le n-1$. We know that $f_{o_io_{i+1}}>\frac{r}{2}$ for $i\in\{j,\ldots,n\}$. If $j=1$, then for any $e\in\alpha_1$, $f_e\ge\frac{r}{2}$. Thus, $\ell_{\alpha_1}\ge\tilde{\ell}_{\alpha_1}$, which contradicts with Lemma \ref{lem:key}. If $j\ge2$, then $f_{o_io_{i+1}}<\frac{r}{2}$ for $i\in\{1,\ldots,j-1\}$. Thus, we derive $f_{d_id_{i+1}}>\frac{r}{2}$ for $i\in\{1,\ldots,j-1\}$. For $e\in\alpha_j$, $f_e\ge\frac{r}{2}$. Thus, $\ell_{\alpha_j}\ge\tilde{\ell}_{\alpha_j}$, which contradicts with Lemma \ref{lem:key}. 

\noindent{\bf Case 2}. $\sigma(2)=\cdots=\sigma(n)=-1$.\\
We relabel the vertices in $C$ as follows: $o_1' = d_1$, $d_1' = o_1$, $o_i' = o_{n+2-i}$, and $d_i' = d_{n+2-i}$ for $i \in \{2,\ldots,n\}$. In this case, the problem is transformed into the scenario where $\sigma(1) = \sigma(2) = \cdots = \sigma(n) = -1$, which is exactly the same as the scenario where $\sigma(1) = \sigma(2) = \cdots = \sigma(n) = 1$. Therefore, this case is proven.

\noindent{\bf Case 3}. $\sigma(2)$, $\sigma(3)$, $\ldots$, $\sigma(n)$ are not all equal and $\sigma(2)=1$.\\
Let $j=\min\{i~|~\sigma(i)=-1\}$. Then $3\le j\le n$. If $\sigma(i)=-1$ for any $i\in\{j,\ldots,n\}$, then $f_{o_{j-1}o_j}=r\ge\frac{r}{2}$. As $\sigma(j-1)=1$ and $\sigma(j)=-1$, by Lemmas \ref{lem:4equal} and \ref{lem:shrink_1}, we derive a new IBP instance with $n-1$ types of travelers, which contradicts with the induction assumption. Otherwise, we only need to consider the case $f_{o_{j-1}o_j}<\frac{r}{2}$. Let $k=\min\{i\ge j+1~|~\sigma(i)=1\}$. Then $\sigma(j)=\cdots=\sigma(k-1)=-1$. By Lemma \ref{lem:ge_le}, we deduce that $f_{o_{j-1}o_j}>f_{o_jo_{j+1}}>\cdots>f_{o_{k-1}o_k}$, which means that $f_{o_{k-1}o_k}<\frac{r}{2}$. As $\sigma(k-1)=-1$ and $\sigma(k)=1$, by Lemmas \ref{lem:4equal} and \ref{lem:shrink_1}, we derive a new IBP instance with $n-1$ types of travelers, which contradicts with the induction assumption. 

\noindent{\bf Case 4}. $\sigma(2)$, $\sigma(3)$, $\ldots$, $\sigma(n)$ are not all equal and $\sigma(2)=-1$.\\
Let $j=\min\{i\ge2~|~\sigma(i)=1\}$. Then $3\le j\le n$. If $\sigma(i)=1$ for any $i\in\{j,\ldots,n\}$, then $f_{o_{j-1}o_j}=r_1<\frac{r}{2}$ by Lemma \ref{lem:r_i}. As $\sigma(j-1)=-1$ and $\sigma(j)=1$, by Lemmas \ref{lem:4equal} and \ref{lem:shrink_1}, we derive a new IBP instance with $n-1$ types of travelers, which contradicts with the induction assumption. Otherwise, we only need to consider the case $f_{o_{j-1}o_j}>\frac{r}{2}$. Let $k=\min\{i\ge j+1~|~\sigma(i)=-1\}$. Then $\sigma(j)=\cdots=\sigma(k-1)=1$. By Lemma \ref{lem:ge_le}, we deduce that $f_{o_{j-1}o_j}<f_{o_jo_{j+1}}<\cdots<f_{o_{k-1}o_k}$, which means that $f_{o_{k-1}o_k}>\frac{r}{2}$. As $\sigma(k-1)=1$ and $\sigma(k)=-1$, by Lemmas \ref{lem:4equal} and \ref{lem:shrink_1}, we derive a new IBP instance with $n-1$ types of travelers, which contradicts with the induction assumption.
\end{proof}

Next, Let us consider a cycle with $n$ OD pairs possessing an almost symmetric distribution of terminal vertices: except for $o_1$ and $d_1$ of type-$1$ travelers, all other $o_i$ and $d_i$ are symmetric where $i\neq1$. Suppose $d_1$ is located between $d_j$ and $d_{j+1}$, where $j\in\{2,\ldots,n\}$ and $d_{n+1}=o_1$, as shown in Figure~\ref{fig:almost-sym-n-od}:

\begin{figure}[h]
\begin{center}
\includegraphics[scale=0.36]{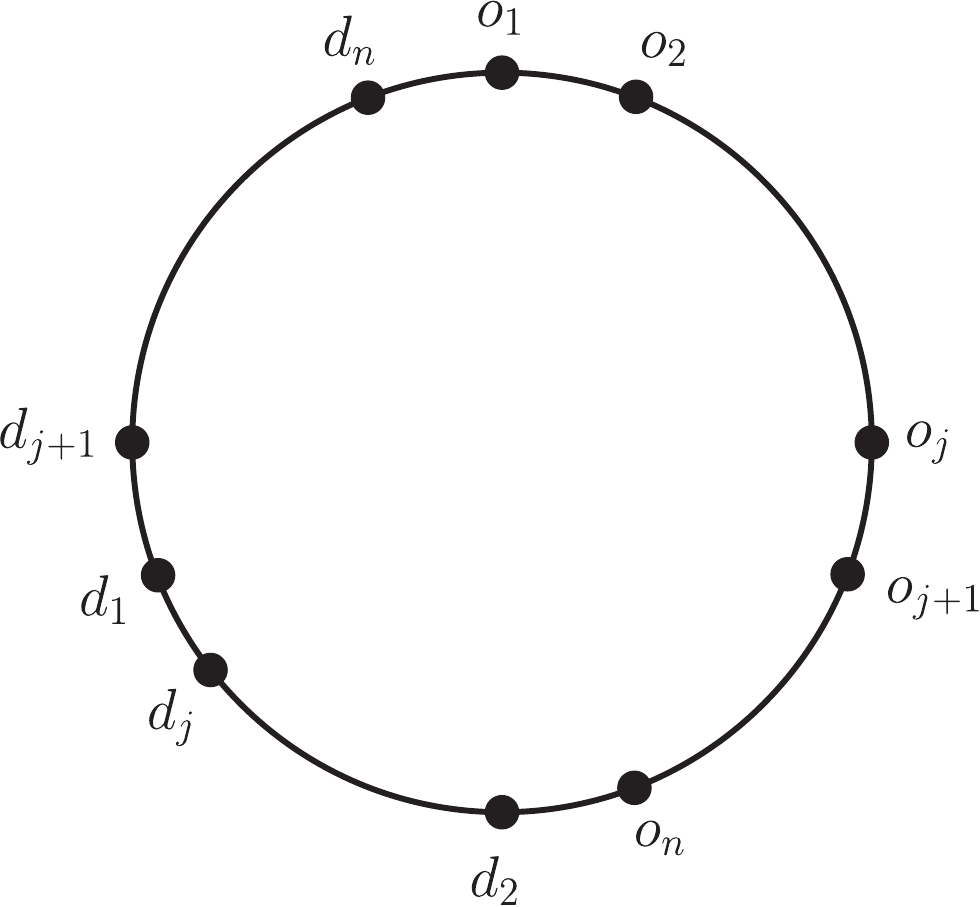}
\caption{The almost symmetric structured network with $n$ OD pairs}\label{fig:almost-sym-n-od}
\end{center}
\end{figure}

\begin{lemma}\label{lem:almost_sym}
A cycle $C$ shown in Figure~\ref{fig:almost-sym-n-od} with $n~(\ge3)$ OD pairs is immune to IBP.
\end{lemma}

\begin{proof}
(by induction on $n$) When $C$ contains three OD pairs, the conclusion holds by Lemma \ref{lem:3od}. Assuming that the conclusion holds for cases containing fewer than $n$ OD pairs, consider an IBP instance where $C$ contains $n$ OD pairs with a completely symmetric distribution of terminal vertices. Suppose that $f$ and $\tilde{f}$ are ICWE flows of the network before and after the information expansion. Let $r=\sum_{i=1}^nr_i$. 
If $\sigma(1)=-1$, since $\ell_{\alpha_1}>\ell_{\beta_1}$ and $\ell_{\alpha_2}\le\ell_{\beta_2}$, then $\sigma(2)=1$. Thus, $\alpha_1\cap\alpha_2=\emptyset$, which means that $\beta_1\cup\beta_2=C$. By Lemma \ref{lem:union_cycle}, we derive a IBP instance with $n-1$ types of travelers, which contradicts with the induction assumption. Thus, $\sigma(1)=1$. Let us discuss the values of $\sigma(2),\ldots,\sigma(n)$ below:

\noindent{\bf Case 1}. $\sigma(2)=\cdots=\sigma(n)=1$.\\
We derive that $f_{o_nd_2}=\sum_{i=1}^nr_i=r$. By Lemma \ref{lem:ge_le}, we have $f_{o_nd_2}>f_{o_{n-1}o_n}>\cdots>f_{o_1o_2}>f_{d_no_1}$ and $f_{o_nd_2}>f_{d_2d_3}>\cdots>f_{d_jd_1}>f_{d_1d_{j+1}}>\cdots>f_{d_{n-1}d_n}>f_{d_no_1}$. Let $k=\min\{i~|~f_{o_io_{i+1}}\ge\frac{r}{2}\}$. Notice that $f_{o_1o_2}=r_1<\frac{r}{2}$ and $f_{o_{n-1}o_n}=\sum_{i=1}^{n-1}r_i>\frac{r}{2}$ by Lemma \ref{lem:r_i}. Then $2\le k\le n-1$.
By the definition of $k$, we know that $f_{o_{k-1}o_k}<\frac{r}{2}$. Then if $k\neq j+1$, then $f_{d_{k-1}d_k}>\frac{r}{2}$; if $k=j+1$, then $f_{d_{k-1}d_1}>\frac{r}{2}$ and $f_{d_1d_{k}}>\frac{r}{2}$. For $e\in\alpha_k$, $f_e\ge\frac{r}{2}$. Thus, $\ell_{\alpha_k}\ge\tilde{\ell}_{\alpha_k}$, which contradicts with Lemma \ref{lem:key}. 

\noindent{\bf Case 2}. $\sigma(2)=\cdots=\sigma(n)=-1$.\\
We derive that $f_{o_1o_2}=\sum_{i=1}^nr_i=r$. By Lemma \ref{lem:ge_le}, $f_{o_1o_2}>f_{o_2o_3}>\cdots>f_{o_{n-1}o_n}>f_{o_nd_2}$ and $f_{o_1o_2}>f_{d_no_1}>\cdots>f_{d_1d_{j+1}}>f_{d_jd_1}>\cdots>f_{d_2d_3}>f_{o_nd_2}$.
Let $k=\max\{i\ge2~|~f_{o_{i-1}o_i}\ge\frac{r}{2}\}$. Notice that $f_{o_2o_3}=r-r_2>\frac{r}{2}$ and $f_{o_nd_2}=r_1<\frac{r}{2}$ by Lemma \ref{lem:r_i}. Then $3\le k\le n$. By the definition of $k$, we know that $f_{o_ko_{k+1}}<\frac{r}{2}$ where $o_{n+1}=d_2$. Then if $k\neq j$, then $f_{d_kd_{k+1}}>\frac{r}{2}$; if $k=j$, then $f_{d_kd_1}>\frac{r}{2}$ and $f_{d_1d_{k+1}}>\frac{r}{2}$. For $e\in\alpha_k$, $f_e\ge\frac{r}{2}$. Thus, $\ell_{\alpha_k}\ge\tilde{\ell}_{\alpha_k}$, which contradicts with Lemma \ref{lem:key}. 

\noindent{\bf Case 3}. $\sigma(2)$, $\sigma(3)$, $\ldots$, $\sigma(n)$ are not all equal and $\sigma(2)=\cdots=\sigma(j)=\sigma(j+1)=-1$.\\
Let $k=\min\{i\ge j+1~|~\sigma(i)=1\}$. Then $\sigma(k-1)=-1$ and $\sigma(k)=1$. If $f_{o_{k-1}o_k}\le\frac{r}{2}$, by Lemmas \ref{lem:4equal} and \ref{lem:shrink_1}, we derive a new IBP instance with $n-1$ types of travelers, which contradicts with the induction assumption. If $f_{o_{k-1}o_k}>\frac{r}{2}$, then there exists $i\ge k+1$ with $\sigma(i)=-1$. Otherwise, $f_{o_{k-1}o_k}=r_1<\frac{r}{2}$ by Lemma \ref{lem:r_i}. Let $l=\min\{i\ge k+1~|~\sigma(i)=-1\}$. Then $\sigma(l-1)=1$ and $\sigma(l)=-1$.  According to Lemma \ref{lem:ge_le}, $f_{o_{l-1}o_l}>f_{o_{k-1}o_k}>\frac{r}{2}$. By Lemmas \ref{lem:4equal} and \ref{lem:shrink_1}, we derive a new IBP instance with $n-1$ types of travelers, which contradicts with the induction assumption. 

\noindent{\bf Case 4}. $\sigma(2)$, $\sigma(3)$, $\ldots$, $\sigma(n)$ are not all equal and $\sigma(2)=\cdots=\sigma(j)=-1$, $\sigma(j+1)=1$.\\
If $f_{o_jo_{j+1}}>\frac{r}{2}$, then there exists $i\ge j+2$ with $\sigma(i)=-1$. Otherwise, $f_{o_jo_{j+1}}=r_1<\frac{r}{2}$ by Lemma \ref{lem:r_i}. Let $l=\min\{i\ge j+2~|~\sigma(i)=-1\}$. Then $\sigma(l-1)=1$ and $\sigma(l)=-1$.  According to Lemma \ref{lem:ge_le}, $f_{o_{l-1}o_l}>f_{o_jo_{j+1}}>\frac{r}{2}$. By Lemmas \ref{lem:4equal} and \ref{lem:shrink_1}, we derive a new IBP instance with $n-1$ types of travelers, which contradicts with the induction assumption. If $f_{o_jo_{j+1}}\le\frac{r}{2}$, then $f_{d_jd_1}\ge\frac{r}{2}$ and $f_{d_1d_{j+1}}\ge\frac{r}{2}$. According to the proof of Lemma \ref{lem:4equal}, we have that $\delta=\ell_{o_jo_{j+1}}=\tilde{\ell}_{o_jo_{j+1}}=\ell_{d_jd_1}+\ell_{d_1d_{j+1}}=\tilde{\ell}_{d_jd_1}+\tilde{\ell}_{d_1d_{j+1}}$. By Lemma \ref{lem:shrink_3}, we derive a new IBP instance with $n-1$ types of travelers, which contradicts with the induction assumption. 

\noindent{\bf Case 5}. $\sigma(2)=-1$, $\sigma(3)$, $\ldots$, $\sigma(n)$ are not all equal and there exists $3\le i\le j$ with $\sigma(i)=1$.\\
Let $k=\min\{3\le i\le j~|~\sigma(i)=1\}$. Then $\sigma(k-1)=-1$ and $\sigma(k)=1$. Thus, $\beta_1\subset\beta_k$. If $f_{o_{k-1}o_k}<\frac{r}{2}$, by Lemmas \ref{lem:4equal} and \ref{lem:shrink_2}, we derive a new IBP instance with $n-1$ types of travelers, which contradicts with the induction assumption. Then $f_{o_{k-1}o_k}\ge\frac{r}{2}$. Let us consider the edge $e\in\alpha_k$ in the following five cases:
\begin{enumerate}
\item
For $e=o_lo_{l+1}$ where $l\in\{k,\ldots,j-1\}$, if $\sigma(l)=-1$ and $\sigma(l+1)=1$, then $f_e>\frac{r}{2}$. Otherwise, by Lemmas \ref{lem:4equal} and \ref{lem:shrink_2}, we derive a new IBP instance with $n-1$ types of travelers, which contradicts with the induction assumption. 
\item
For $e=o_jo_{j+1}$, if $\sigma(j)=-1$ and $\sigma(j+1)=1$, then $f_e>\frac{r}{2}$. Otherwise, by Lemmas \ref{lem:4equal_2} and \ref{lem:shrink_3}, we derive a new IBP instance with $n-1$ types of travelers, which contradicts with the induction assumption.
\item
For $e=o_lo_{l+1}$ where $l\in\{j+1,\ldots,n-1\}$, if $\sigma(l)=-1$ and $\sigma(l+1)=1$, then $f_e>\frac{r}{2}$. Otherwise, by Lemmas \ref{lem:4equal} and \ref{lem:shrink_1}, we derive a new IBP instance with $n-1$ types of travelers, which contradicts with the induction assumption. 
\item
For $e=o_nd_2$, if $\sigma(n)=-1$ and $\sigma(2)=-1$, then $f_e>\frac{r}{2}$. Otherwise, by Lemmas \ref{lem:4equal_2} and \ref{lem:shrink_3}, we derive a new IBP instance with $n-1$ types of travelers, which contradicts with the induction assumption.
\item
For $e=d_ld_{l+1}$ where $l\in\{2,\ldots,k-1\}$, if $\sigma(l)=1$ and $\sigma(l+1)=-1$, then $f_e>\frac{r}{2}$. Otherwise, by Lemmas \ref{lem:4equal} and \ref{lem:shrink_2}, we derive a new IBP instance with $n-1$ types of travelers, which contradicts with the induction assumption.
\end{enumerate}
Let us take the edge with the minimum  flow in $\alpha_k$, denoted as $e^\ast$. We aim to prove that $f_{e^\ast}>\frac{r}{2}$. If $e^\ast=o_ko_{k+1}$, since $\sigma(k)=1$, then $f_{o_ko_{k+1}}>f_{o_{k-1}o_k}\ge\frac{r}{2}$. Since $\sigma(k-1)=-1$, $f_{d_{k-1}d_k}>f_{d_{k-2}d_{k-1}}$. Thus, $e^\ast\neq d_{k-1}d_k$. When $e^\ast\neq o_ko_{k+1}$ and $e^\ast\neq d_{k-1}d_k$, $e^\ast$ has edges on both sides in $\alpha_k$ and $f_{e^\ast}$ does not exceed the flow of the edges on its two sides, then $e^\ast$ must belong to one of the five cases mentioned above, which means that $f_{e^\ast}>\frac{r}{2}$. Therefore, $\ell_{\alpha_k}\ge\tilde{\ell}_{\alpha_k}$, which contradicts with Lemma \ref{lem:key}.

\noindent{\bf Case 6}. $\sigma(2)$, $\sigma(3)$, $\ldots$, $\sigma(n)$ are not all equal and $\sigma(2)=\sigma(n)=1$.\\
Then, $\beta_1\subset\beta_2$. If $f_{o_1o_2}<\frac{r}{2}$, since $\sigma(1)=1$, then $f_{d_no_1}<f_{o_1o_2}<\frac{r}{2}$. Since $\sigma(2)=1$ and $\sigma(n)=1$, by Lemmas \ref{lem:4equal_2} and \ref{lem:shrink_3}, we derive a new IBP instance with $n-1$ types of travelers, which contradicts with the induction assumption. Then $f_{o_1o_2}\ge\frac{r}{2}$. Let us consider the edge $e\in\alpha_2$ in the following three cases: 
\begin{enumerate}
\item
For $e=o_lo_{l+1}$ where $l\in\{2,\ldots,j-1\}$, if $\sigma(l)=-1$ and $\sigma(l+1)=1$, then $f_e>\frac{r}{2}$. Otherwise, by Lemmas \ref{lem:4equal} and \ref{lem:shrink_2}, we derive a new IBP instance with $n-1$ types of travelers, which contradicts with the induction assumption. 
\item
For $e=o_jo_{j+1}$, if $\sigma(j)=-1$ and $\sigma(j+1)=1$, then $f_e>\frac{r}{2}$. Otherwise, by Lemmas \ref{lem:4equal_2} and \ref{lem:shrink_3}, we derive a new IBP instance with $n-1$ types of travelers, which contradicts with the induction assumption.
\item
For $e=o_lo_{l+1}$ where $l\in\{j+1,\ldots,n-1\}$, if $\sigma(l)=-1$ and $\sigma(l+1)=1$, then $f_e>\frac{r}{2}$. Otherwise, by Lemmas \ref{lem:4equal} and \ref{lem:shrink_1}, we derive a new IBP instance with $n-1$ types of travelers, which contradicts with the induction assumption. 
\end{enumerate}
Let us take the edge with the minimum  flow in $\alpha_2$, denoted as $e^\ast$. We aim to prove that $f_{e^\ast}>\frac{r}{2}$. If $e^\ast=o_2o_3$, since $\sigma(2)=1$, then $f_{o_2o_3}>f_{o_1o_2}\ge\frac{r}{2}$. Since $\sigma(n)=1$, $f_{o_nd_2}>f_{o_{n-1}o_n}$. Thus, $e^\ast\neq o_nd_2$. When $e^\ast\neq o_2o_3$ and $e^\ast\neq o_nd_2$, $e^\ast$ has edges on both sides in $\alpha_2$ and $f_{e^\ast}$ does not exceed the flow of the edges on its two sides, then $e^\ast$ must belong to one of the three cases mentioned above, which means that $f_{e^\ast}>\frac{r}{2}$. Therefore, $\ell_{\alpha_2}\ge\tilde{\ell}_{\alpha_2}$, which contradicts with Lemma \ref{lem:key}.
 
\noindent{\bf Case 7}. $\sigma(2)=1$ and $\sigma(j)=\sigma(n)=-1$.\\
Let us renumber the vertices on $C$ as follows: 
\begin{itemize}
\item
Swap the labels of $o_1$ and $d_1$.
\item
Swap the labels of $o_k$ and $d_{j+2-k}$ for $k \in \{2, \ldots, j\}$.
\item
Swap the labels of $o_k$ and $o_{n+j+1-k}$ for $k \in \{j+1, \ldots, j+\lfloor\frac{n-j}{2}\rfloor\}$.
\item
Swap the labels of $d_k$ and $d_{n+j+1-k}$ for $k \in \{j+1, \ldots, j+\lfloor\frac{n-j}{2}\rfloor\}$.
\end{itemize}
Then $\sigma(1)=-1$ and $\sigma(2)=1$, which is equivalent to the scenario where $\sigma(1)=1$ and $\sigma(2)=-1$. The result in this scenario has been proved in Cases 2, 3, 4 and 5.  

\noindent{\bf Case 8}. $\sigma(2)=\sigma(j)=1$ and  $\sigma(j+1)=\sigma(n)=-1$.\\
Let us renumber the vertices on $C$ as follows: 
\begin{itemize}
\item
Swap the labels of $o_1$ and $d_1$.
\item
Swap the labels of $o_k$ and $d_{j+2-k}$ for $k \in \{2, \ldots, j\}$.
\item
Swap the labels of $o_k$ and $o_{n+j+1-k}$ for $k \in \{j+1, \ldots, j+\lfloor\frac{n-j}{2}\rfloor\}$.
\item
Swap the labels of $d_k$ and $d_{n+j+1-k}$ for $k \in \{j+1, \ldots, j+\lfloor\frac{n-j}{2}\rfloor\}$.
\end{itemize}
Then $\sigma(1)=\sigma(2)=\sigma(n)=-1$, which is equivalent to the scenario where $\sigma(1)=\sigma(2)=\sigma(n)=1$. The result in this scenario has been proved in Cases 1 and 6. 

\noindent{\bf Case 9}. $\sigma(2)=\sigma(j)=\sigma(j+1)=1$ and $\sigma(n)=-1$.\\
Then, $\beta_1\subset\beta_2$. We have the following inequalities:
\begin{itemize}
\item
$f_{d_{n-1}d_n}<f_{d_no_1}<f_{o_1o_2}$.
\item
$f_{o_{j-1}o_j}<f_{o_jo_{j+1}}$.
\item
$f_{d_2d_3}<f_{o_nd_2}$.
\item
$f_{d_{j+1}d_{j+2}}<f_{d_1d_j+1}<f_{d_jd_1}$.
\end{itemize}
Let us consider the edge with the minimum flow on $C$ and denote it as $e^\ast$. Based on the previous inequalities, we can conclude that $e^\ast$ does not belong to the set $\{d_no_1, o_1o_2, o_jo_{j+1}, o_nd_2, $ $d_jd_1, d_1d_{j+1}\}$. Therefore, there are only two possible cases:
\begin{itemize}
\item
For $e^\ast=o_ko_{k+1}$ where $k\in\{2,\ldots,n-1\}\setminus\{j\}$, then $\sigma(k)=-1$ and $\sigma(k+1)=1$. If $f_{e^\ast}\le\frac{r}{2}$, by Lemmas \ref{lem:4equal} and \ref{lem:shrink_2}, we derive a new IBP instance with $n-1$ types of travelers, which contradicts with the induction assumption. 
\item
For $e^\ast=d_kd_{k+1}$ where $k\in\{2,\ldots,n-1\}\setminus\{j\}$, then $\sigma(k)=1$ and $\sigma(k+1)=-1$. If $f_{e^\ast}\le\frac{r}{2}$, by Lemmas \ref{lem:4equal} and \ref{lem:shrink_2}, we derive a new IBP instance with $n-1$ types of travelers, which contradicts with the induction assumption. 
\end{itemize}
Thus, $f_{e}>\frac{r}{2}$ for any $e\in\alpha_2$. Therefore, $\ell_{\alpha_2}\ge\tilde{\ell}_{\alpha_2}$, which contradicts with Lemma \ref{lem:key}.

\noindent In conclusion, a cycle $C$ shown in Figure~\ref{fig:almost-sym-n-od} with $n~(\ge3)$ OD pairs is IBP-free.
\end{proof}

Next, we will transform a cycle containing $n$ OD pairs with the general distribution of the terminal vertices into a completely symmetric structured network as shown in Figure~\ref{fig:sym-n-od} or an almost symmetric structured network as shown in Figure~\ref{fig:almost-sym-n-od} through a finite number of steps. 
The following lemma allows us to convert the traffic rates of all types of travelers into integers.

\begin{lemma}\label{lem:r_to_int}
If there exists an IBP instance in a cycle with $n~(\ge3)$ OD pairs, then we can construct an IBP instance in a cycle with $n~(\ge3)$ OD pairs and all the traffic rates are integers.
\end{lemma}
\begin{proof}
If all the traffic rates are rational numbers, then we assume that $r_i=\frac{p_i}{q_i}$ where $p_i,q_i$ are positive integers for any $i\in[n]$. Let $Q$ be the least common multiple of $q_1,q_2,\ldots,q_n$. Then the new IBP instance have the traffic rates: $r_i'=r_i\cdot Q$ for any $i\in[n]$.  
Otherwise, suppose that $r_i=a$ is an irrational number for some $i\in[n]$. Let $A=\{e\in C~|~\ell_e=\tilde{\ell}_e\}$ and $b=\min_{e\in C\setminus A}|f_e-\tilde{f}_e|$. Since the rational points are dense on the number axis, we can choose a rational number $c$ such that $a<c<a+\frac{b}{2}$. Let $r_i'=c$ and $r_j'=r_j$ for $j\in[n]\setminus\{i\}$. Let $f'$ be the flow of the network before the information expansion where type-$i$ travelers unanimously choose the path $\alpha_i$ for any $i\in[n]$ and $\tilde{f}'$ be the flow of the network before the information expansion where type-$i$ travelers unanimously choose the path $\beta_i$ for any $i\in[n]$. Since $a<c<a+\frac{b}{2}$, we have that $|f_e-f_e'|<\frac{b}{2}$ and $|\tilde{f}_e-\tilde{f}'_e|<\frac{b}{2}$ for any $e\in C$. Then, we derive
\begin{itemize}
\item
If $f_e-\tilde{f}_e>0$, then $f_e-\tilde{f}_e\ge b$. We have that $f'_e-\tilde{f}'_e=(f_e-\tilde{f}_e)-(f_e-f'_e)-(\tilde{f}'_e-\tilde{f}_e)$ $>b-\frac{b}{2}-\frac{b}{2}=0$.
\item
If $f_e-\tilde{f}_e<0$, then $f_e-\tilde{f}_e\le -b$. We have that $f'_e-\tilde{f}'_e=(f_e-\tilde{f}_e)-(f_e-f'_e)-(\tilde{f}'_e-\tilde{f}_e)$ $<-b+\frac{b}{2}+\frac{b}{2}=0$.
\end{itemize}
Next, we construct the latency functions for any edge $e\in C$:
For $e\in A$, $\hat{\ell}_e(x)\equiv \ell_e$ where $x\in[0,\infty)$. For $e\in C\setminus A$,
\[
\hat{\ell}_e(x)=\begin{cases}
\min\{\ell_e,\tilde{\ell}_e\} & x\in[0,\min\{f'_e,\tilde{f}'_e\}],\\
\frac{\ell_e-\tilde{\ell}_e}{f_e'-\tilde{f}_e'}x+\frac{\tilde{\ell}_e\cdot f_e'-\ell_e\cdot\tilde{f}_e'}{f_e'-\tilde{f}_e'} & x\in(\min\{f'_e,\tilde{f}'_e\},\max\{f'_e,\tilde{f}'_e\}),\\
\max\{\ell_e,\tilde{\ell}_e\} & x\in[\max\{f'_e,\tilde{f}'_e\},+\infty).
\end{cases}
\]
Then, it follows that $\ell'_e(f_e')=\ell_e$ and $\ell'_e(\tilde{f}'_e)=\tilde{\ell}_e$ for any $e\in C$. Then $f'$ and $\tilde{f}'$ are ICWE flows of the network before and after the information expansion. Thus, this is a new IBP instance. 
We can repeatedly perform the above operation, converting all $r_i$ into rational numbers, and then further transform all $r_i$ into integers.
\end{proof}

The next lemma allows us to transform all integers $r_i$ into $1$.
\begin{lemma}\label{lem:r_to_1}
If there exits an IBP instance in a cycle with $n~(\ge3)$ OD pairs and all integer traffic rates, then we can construct an IBP instance in a cycle with $(\sum_{i=1}^n r_i)$ OD pairs and all the traffic rates are equal to $1$.
\end{lemma}
\begin{proof}
We divide the type-$i$ travelers into $r_i$ types of travelers: $(i,1),(i,2),$ $\ldots,(i,r_i)$. These types of travelers have a traffic rate of $1$ and their OD pairs and the path choices remain consistent with the type-$i$ travelers. On the cycle $C$, we replace the vertex $o_i$ with a path of length $r_i-1$, denoted as $o_i^1 o_i^2\cdots o_i^{r_i}$, and replace the vertex $d_i$ with a path of length $r_i-1$, denoted as $d_i^1 d_i^2\cdots d_i^{r_i}$. The latency of all newly added edges is defined as $0$.
\end{proof}

For any IBP instance with a traffic rate of $1$, we can transform it into an IBP instance of the completely symmetric structured cycle or the almost symmetric structured cycle.

\begin{lemma}\label{lem:gen_to_sym}
If there exits an IBP instance in a cycle with $n~(\ge3)$ OD pairs and all traffic rates are $1$, then we can construct an IBP instance in the completely symmetric structured cycle or the almost symmetric structured cycle.
\end{lemma}

\begin{proof}
Let $K=[n]\setminus\{1\}$ and $M=\{\{i,j\}~|~\alpha_j\subset\alpha_i,~i,j\in K\}$. Suppose that $M\neq\emptyset$. Take $\{i,j\}\in M$ and assume that $\alpha_j\subset\alpha_i$ as shown in Figure~\ref{fig:j-sub-i}. $o_i$, $o_j$, $d_j$, and $d_i$ divide the cycle $C$ into four paths:
\begin{itemize}
\item
$P_1$ is the path from $o_i$ clockwise to $o_j$.
\item
$P_2$ is the path from $o_j$ clockwise to $d_j$, which is actually $\alpha_j$.
\item
$P_3$ is the path from $d_j$ clockwise to $d_i$.
\item
$P_4$ is the path from $d_i$ clockwise to $o_i$, which is actually $\beta_i$.
\end{itemize}

\begin{figure}[h]
\begin{center}
\includegraphics[scale=0.36]{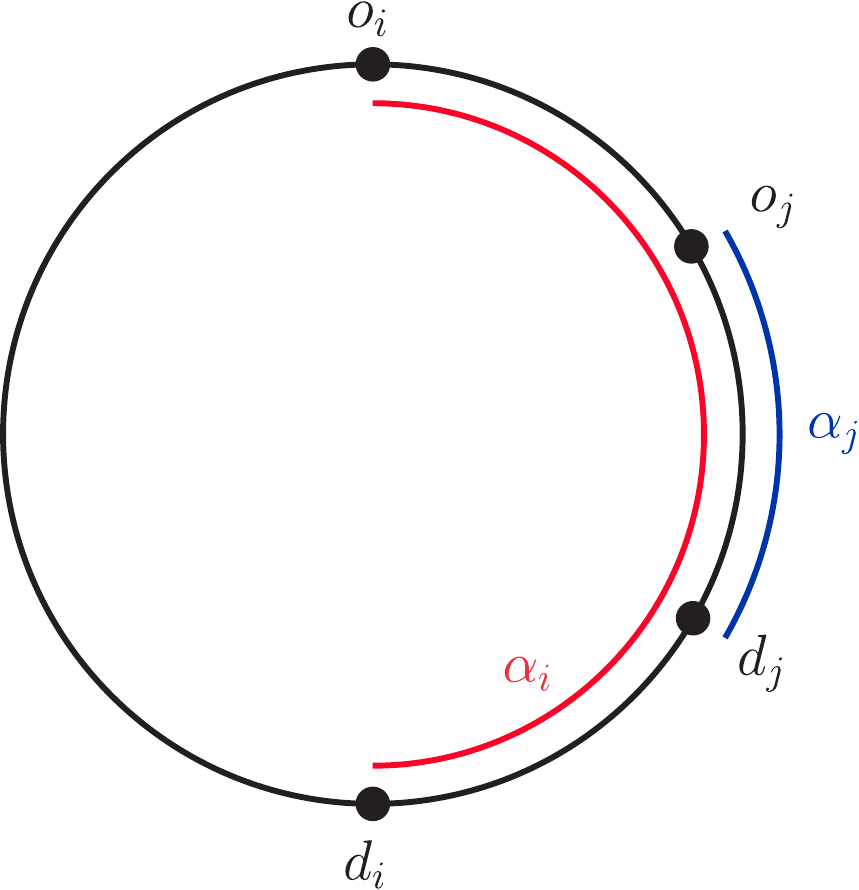}
\caption{The schematic diagram of $\alpha_j\subset\alpha_i$}\label{fig:j-sub-i}
\end{center}
\end{figure}

Swap $d_i$ and $d_j$. Define new paths $\bar{\alpha}_i = P_1 \cup P_2$, $\bar{\beta}_i = P_3 \cup P_4$, $\bar{\alpha}_j = P_2 \cup P_3$, and $\bar{\beta}_j = P_1 \cup P_4$. For $k\in\{i,j\}$, type-$k$ travelers choose $\bar{\alpha}_k$ before and $\bar{\beta}_k$ after the information expansion. We can observe that the traffic flow on each edge of the cycle remains unchanged after the swap, indicating that we still have an IBP instance. Next, we will prove that with each iteration of the above operation, the value of $|M|$ strictly decreases.

Let $M'$ be the set obtained after performing one swap operation on $M$. We firstly have that $\{i,j\}\in M$ and $\{i,j\}\not\in M'$. Consider any $\{k,l\}\in M'$ and $\{k,l\}\not\in M$:
\begin{itemize}
\item
If $\{k,l\}\cap\{i,j\}=\emptyset$, since $\{k,l\}\in M'$, then $\{k,l\}\in M$, which is a contradiction.
\item
If $k\neq i,j$ and $l=i$, since $\{k,i\}\in M'$ and $\{k,i\}\not\in M$, we have either $o_k\in P_3$ and $d_k\in P_4$, or $o_k\in P_4$ and $d_k\in P_3$. Thus, $\{k,j\}\in M$ and $\{k,j\}\not\in M'$.
\item
If $k\neq i,j$ and $l=j$, since $\{k,j\}\in M'$ and $\{k,j\}\not\in M$, we have either $o_k\in P_2$ and $d_k\in P_3$, or $o_k\in P_3$ and $d_k\in P_2$. Thus, $\{k,i\}\in M$ and $\{k,i\}\not\in M'$.
\end{itemize}
Thus, it follows that $|M'|<|M|$.

Therefore, after a finite number of steps, an IBP instance obtained corresponds to $M=\emptyset$, i.e., for any $i,j\in K$, $\alpha_i\setminus\alpha_j\neq\emptyset$. If there exist $i,j\in K$ such that $\alpha_i\cap\alpha_j=\emptyset$ or $\alpha_i\cup\alpha_j=C$, according to Lemma \ref{lem:union_cycle}, we can construct a new IBP instance by removing one type of travelers. Therefore, the final IBP instance satisfies: for any $i,j\in K$, $\alpha_i\setminus\alpha_j\neq\emptyset$, $\alpha_i\cap\alpha_j\neq\emptyset$ and $\alpha_i\cup\alpha_j\neq C$, which is actually an IBP instance of the completely symmetric structured cycle or the almost symmetric structured cycle.
\end{proof}

Finally, we are ready to prove the main theorem of this section, which establishes our characterization for IBP-free networks with multiple OD pairs (Theorem~\ref{thm:main_0}).

\begin{proof}[Proof of Theorem~\ref{thm:cycle_main}]
By Lemma \ref{lem:2od}, a cycle with two OD pairs is IBP-free.
Suppose that there exists an IBP instance in $C$ with $n~(\ge3)$ OD pairs. By Lemmas \ref{lem:r_to_int}, \ref{lem:r_to_1} and \ref{lem:gen_to_sym}, we obtain an IBP instance of the completely symmetric structured cycle or the almost symmetric structured cycle, which contradicts with Lemma \ref{lem:symmetric} or Lemma \ref{lem:almost_sym}.
\end{proof}
\end{document}